\documentclass[pdflatex,sn-apa]{sn-jnl}

\usepackage{graphicx}%
\usepackage{multirow}%
\usepackage{amsmath,amssymb,amsfonts}%
\usepackage{amsthm}%
\usepackage{mathrsfs}%
\usepackage[title]{appendix}%
\usepackage{xcolor}%
\usepackage{textcomp}%
\usepackage{manyfoot}%
\usepackage{booktabs}%
\usepackage{algorithm}%
\usepackage{algorithmicx}%
\usepackage{algpseudocode}%
\usepackage{listings}%

\usepackage{multicol}
\usepackage{multirow}
\usepackage{amssymb}
\usepackage{url}
\usepackage{enumerate}
\usepackage{amsthm}
\usepackage{color}
\usepackage{comment}
\usepackage{dsfont}
\usepackage{natbib}
\usepackage[misc]{ifsym}
\usepackage{enumitem}  

\usepackage{actuarialsymbol}
\usepackage{actuarialangle}

\usepackage{color,xcolor,comment}
\usepackage{bm}

\renewcommand*{\P}{\mathbb{P}}

\newcommand*{\VaR}{\operatorname{VaR}}

\newcommand*{\ES}{\operatorname{ES}}

\newcommand{\E}{\mathbb{E}}

\newcommand{\id}{\mathds{1}}

\renewcommand{\ge}{\geqslant}
\renewcommand{\le}{\leqslant}

\renewcommand{\leq}{\leqslant}
\renewcommand{\epsilon}{\varepsilon}

\newcommand{\rd}{\mathrm{d}}

\theoremstyle{thmstyleone}%
\newtheorem{Theorem}{Theorem}
\newtheorem{Proposition}[Theorem]{Proposition}%
\newtheorem{Corollary}[Theorem]{Corollary}%
\newtheorem{Lemma}[Theorem]{Lemma}%

\theoremstyle{thmstyletwo}%
\newtheorem{Remark}{Remark}%

\theoremstyle{thmstylethree}%
\newtheorem{Definition}{Definition}%

\begin{document}

\title[Robust risk evaluation of joint life insurance under dependence uncertainty]{Robust risk evaluation of joint life insurance\\ under dependence uncertainty}


\author*[1]{\fnm{Takaaki} \sur{Koike}}\email{takaaki.koike@r.hit-u.ac.jp}



\affil*[1]{\orgdiv{Graduate School of Economics}, \orgname{Hitotsubashi University}, \orgaddress{\street{2-1, Naka}, \city{Kunitachi}, \postcode{186-8601}, \state{Tokyo}, \country{Japan}}}




\abstract{
Dependence among multiple lifetimes is a key factor for pricing and evaluating the risk of joint life insurance products. The dependence structure can be
exposed to model uncertainty when available data and information are limited.
We address robust pricing and risk evaluation of joint life insurance products against dependence uncertainty between two lifetimes. 
We first show that, for some class of standard contracts, the risk evaluation based on a distortion risk measure is monotone with respect to the concordance order of the underlying copula.
Based on this monotonicity, we then study the most conservative and anti-conservative risk evaluations for this class of contracts. 
We prove that the bounds for the mean, Value-at-Risk and Expected Shortfall are computed by combinations of linear programs when the uncertainty set is defined by a norm-ball centered around a reference copula.
Our numerical analysis reveals that the sensitivity of the risk evaluation against the choice of the copula differs depending on the risk measure and the type of the contract, and our proposed bounds can improve the existing bounds based on the available information.
}




\maketitle

\section{Introduction}\label{sec:introduction}

Insurance and pension products covering multiple insureds, typically a husband and wife in a family, serve an important social need. 
For a joint life insurance product involving two insureds, 
let $X\sim F$ and $Y\sim G$ be their (continuous) lifetimes, respectively, at the start of the contract.
The first-to-die life insurance, for example, provides a death benefit at time $\min(X,Y)$. 
The second-to-die life insurance, also known as the survivorship life insurance,  pays out at time $\max(X,Y)$.
Pricing and risk evaluation of such products involve modeling not only each of $X$ and $Y$ but also the dependence between them.
Let $C$ be the survival copula of $(X,Y)$, that is, the cumulative distribution function (cdf) of $(U,V)=(\bar F(X),\bar G(Y))$, where $\bar F=1-F$ and $\bar G=1-G$.
The reader is referred to~\citet{nelsen2006introduction} for an introduction to copula theory.
Copulas are a prevalent tool for modeling joint life insurance products; see, for example,~\citet{frees1996annuity,youn1999statistical,carriere2000bivariate,youn2001pricing,shemyakin2006copula,dufresne2018age,gobbi2019joint}.

Although independence between $X$ and $Y$ is a common assumption as in~\citet{dickson2019actuarial}, it can be unrealistic due to, for example, the shared lifestyle, common disasters affecting a family and the so-called broken-heart syndrome~\citep{parkes1969broken}, all of which lead to positive dependence among lifetimes in a family.
The reader is referred to~\citet{denuit1999multilife,youn1999statistical,hougaard2000analysis,denuit2001measuring,spreeuw2006types,spreeuw2013investigating,lu2017broken}
for further analyses and discussions on the lifetime dependence.
Although some advanced models are available, the dependence structure is typically more exposed to model uncertainty compared with the marginal distributions due to, for example, censoring of data, complicated asymmetric and non-exchangeable dependence structure between lifetimes and non-stationarity of the model over a long policy term; see also~\citet{kaas2009worst,embrechts2013model,bignozzi2015reducing,bernard2017risk,liu2017collective,ruschendorf2024model} for risk analyses under dependence uncertainty.
In such a case, over- and under-estimation of the price and the risk are of main concerns from the viewpoints of price competition and conservative risk evaluation.
Moreover, if some assumption is imposed on the dependence among lifetimes in actuarial practice, actuaries may be concerned with whether this assumption tends to over-estimate or under-estimate the risk for a specific insurance product.

Motivated by these concerns, we address the calculation of the bounds on the price and risk of a joint life insurance product under various types of dependence uncertainty between two lifetimes.
Let $L$ be the random present value of the expense of a contract for the insurer.
We quantify the price or risk of this contract by $\varrho(L)$, where $\varrho$ is a distortion risk measure, which includes the mean $\mathbb{E}[L]$, \emph{Value-at-Risk (VaR)}: $\operatorname{VaR}_\alpha(L)=\inf\{l \in \mathbb{R}: F_L(l)\ge \alpha\}$ and \emph{Expected Shortfall (ES)}: $\operatorname{ES}_\alpha(L)=(1/(1-\alpha))\int_{\alpha}^1 \operatorname{VaR}_\beta(L)\rd \beta$ for $\alpha \in (0,1)$ typically close to $1$. 
We call $\mu(L)=\mathbb{E}[L]$ the price of this contract.
To clarify the point of this study, we assume that marginal distributions of $X$ and $Y$ are known, and the survival copula $C$ of $(X,Y)$ is left unspecified within a class of copulas $\mathcal D\subseteq \mathcal C$ called the \emph{uncertainty set}.
Here we denote by $\mathcal C$ the set of all bivariate copulas.
Under this assumption, the quantity of interest $\varrho(L)$ is a function of $C$, denoted by $\varrho(C)$ with abuse of notation, and our goal is to calculate
\begin{align*}
\underline \varrho(\mathcal D)=\inf\{\varrho(C): C\in \mathcal D\}\quad\text{and}\quad
\overline \varrho(\mathcal D)=\sup\{\varrho(C): C\in \mathcal D\}
\end{align*}
either analytically or numerically.
Calculating these bounds is also relevant to quantifying model risk as studied in~\citet{barrieu2015assessing}.

In Section~\ref{sec:preliminaries}, we present the general form of joint life insurance contracts considered in this paper, which covers many standard ones as listed in Appendix~\ref{sec:contracts}.
It is shown in~\citet{denuit1999multilife,denuit2001measuring} that the price of some standard contracts is monotone with respect to (w.r.t.) what is called the \emph{concordance order} on $\mathcal C$, and thus, for such contracts, the bounds on the price over $\mathcal C$ are attained by the well-known \emph{Fr{\'e}chet-Hoeffding bounds}~\citep{nelsen2006introduction}.
The objective of our study is to improve and extend this result in various directions.
Notably, since the price $\mu(L)$ does not quantify the extreme risk of the payoff $L$, we also consider distortion risk measures such as VaR and ES. 
In Section~\ref{sec:monotonicity:distortion}, we show monotonicity of distortion risk measures w.r.t. the concordance order.
Since $\mathcal C$ may be too large and the attaining copulas may be unrealistic, this monotonicity property allows us to improve bounds on $\varrho$ over various subsets of $\mathcal C$ which reflect partial information available on the dependence between $X$ and $Y$; see Remark~\ref{rem:tankov}.
Since these improved bounds are still not practical in typical situations, Section~\ref{sec:uncertainty:level} considers the collection of copulas whose distance from the prescribed \emph{reference copula} $C^{\operatorname{ref}}\in \mathcal C$ is no larger than $\epsilon>0$.
Taking this set corresponds to the situation where the underlying copula of $(X,Y)$ is specified by $C^{\operatorname{ref}}$ with the level of uncertainty $\epsilon$.
This uncertainty set may be useful when, for example, 
the copula $C^{\operatorname{ref}}$ is taken from the literature due to the scarcity of available data, but $C^{\operatorname{ref}}$ is fitted to the population with different attributes, such as generation and geographic region.
We show that these bounds can be derived by combinations of linear programs when the distance is measured by the $\mathcal L^1$ or $\mathcal L^\infty$-distance between two copulas evaluated on a relevant region.
 In Section~\ref{sec:numerical}, we numerically demonstrate the benefits of these bounds for illustrating the sensitivity of the risk evaluation against dependence uncertainty, and for improving the existing bounds based on the available information.
 
The structure of this paper is as follows. 
In Section~\ref{sec:preliminaries}, we present notation and the general form of joint life insurance contracts to be considered in this paper.
Section~\ref{sec:monotonicity:distortion} shows monotonicity of distortion risk measures for these contracts w.r.t. the concordance order of the underlying copula.
Section~\ref{sec:uncertainty:level} shows bounds for the case when $\mathcal D$ is the set of copulas close to a prespecified reference copula.
In Section~\ref{sec:numerical}, we conduct numerical experiments to compare various bounds.
Finally, we conclude this study in Section~\ref{sec:conclusion} with potential directions for future research.
We defer all the proofs to Appendix~\ref{sec:proofs}.

\section{Preliminaries}\label{sec:preliminaries}

This section introduces notation and the general form of standard joint life insurance contracts to be considered in this paper.

For continuous lifetimes $X\sim F$ and $Y\sim G$ of two insureds at the start of the contract, let
$T_{\wedge}=\operatorname{min}(X,Y)$ and $T_{\vee}=\operatorname{max}(X,Y)$.
For $t\ge 0$, we have that 
\begin{align*}
  \mathbb{P}(T_{\wedge}\ge t)&= C\left(
    \bar F(t),\bar G(t)
    \right),\\
       \mathbb{P}(T_{\vee}\ge t)&=\bar F(t)+\bar G(t)-
         C\left(
    \bar F(t),\bar G(t)
    \right),
\end{align*}
where $C$, the survival copula of $(X,Y)$, is the cdf of $(\bar F(X),\bar G(Y))$.
Let $K_X=\lfloor X \rfloor$, $K_Y=\lfloor Y \rfloor$,  $K_{\wedge}=\lfloor T_{\wedge} \rfloor$ and $K_{\vee}=\lfloor T_{\vee} \rfloor$ be curtate lifetimes, where $\lfloor \cdot \rfloor$ is the floor function.
For every $t \in \mathbb{N}_0=\{0,1,\dots\}$, we have that
\begin{align*}
    \mathbb{P}(K_{\wedge}\ge t)=\mathbb{P}(T_{\wedge}\ge t)\quad\text{and}\quad
    \mathbb{P}(K_{\vee}\ge t)=\mathbb{P}(T_{\vee}\ge t),
\end{align*}
and thus the laws of $K_{\wedge}$ and $K_{\vee}$ are completely determined by $\bar F$, $\bar G$ and $C$.

Let $L$ be a $[0,\infty)$-valued random variable, interpreted as the present value of the expense of a contract for the insurer.
In this paper, we consider contracts such that the payoff $L$ is determined by $(K_{\wedge},K_{\vee})$.
To simplify the setting, we assume that quantities other than $(K_{\wedge},K_{\vee})$, such as an effective (annual) rate of interest (and thus the discount factor), are non-random constants.
This form of payoff is admitted for many standard contracts whose payments are on annual basis; see Appendix~\ref{sec:contracts} for the list of such contracts.

We quantify the risk of a contract with payoff $L$ by the \emph{distortion risk measure}:
\begin{align*}
\varrho_h(L)=\int_0^{\infty} h(\mathbb{P}(L>x))\rd x,
\end{align*}
provided that it is well-defined, where
the \emph{distortion function} $h:[0,1]\rightarrow [0,1]$ is an increasing function satisfying $h(0)=0$ and $h(1)=1$.
Distortion risk measures are axiomatically characterized in the context of insurance pricing; see~\citet{wang1997axiomatic}.
Examples of distortion risk measures include the mean $\mu(L)=\mathbb{E}[L]$ with $h^{\mathbb{E}}(\beta)=\beta$, VaR with $h_\alpha^{\operatorname{VaR}}(\beta)=\id_{\{\beta > 1-\alpha\}}$ for $\alpha \in (0,1)$ and ES with  $h_\alpha^{\operatorname{ES}}(\beta)=\min(1,\beta/(1-\alpha))$ for $\alpha \in (0,1)$.

As $\varrho_h(L)$ is a function of $C$ in our setting, we denote it by $\varrho_h(C)$.
For two copulas $C_1,C_2\in \mathcal C$, we say that $C_2$ is more \emph{concordant} than $C_1$, denoted by $C_1 \preceq C_2$, if $C_1(u,v)\le C_2(u,v)$ for all $(u,v)\in [0,1]^2$.
As shown in~\citet{denuit1999multilife,denuit2001measuring}, the price $C\mapsto \mu(C)$ is  monotone w.r.t.~$\preceq$ for many standard contracts; see Appendix~\ref{sec:contracts}.

\section{Monotonicity of distortion risk measures}\label{sec:monotonicity:distortion}

In this section, we show that the risk measured by $\varrho_h(C)$ is monotone w.r.t. $\preceq$, the order on the degree of dependence between $X$ and $Y$, when the payoff is of a specific form. To this end, we first define \emph{monotonicity} of a payoff $L$ of a contract.
An analogous notion for the case of contracts with a single insured can be found in ~\citet{pichler2014insurance}.
Note that the terms `increasing' and `decreasing' are used in a non-strict sense throughout the paper.

\begin{Definition}
[Monotonicity of payoff]
    A random payoff $L$ is called monotone if $L=g(K_{\wedge})$ or $L=g(K_{\vee})$ for some monotone function $g:\mathbb{N}_0\rightarrow[0,\infty)$.
    In particular, if $L=g(K)$ for an increasing (decreasing) function $g$, where $K$ is $K_{\wedge}$ or $K_{\vee}$, then we say that $L$ is increasing (decreasing) in $K$.
\end{Definition}

This class of contracts covers many standard ones described in Appendix~\ref{sec:contracts}, such as joint-life and last-survivor insurances and annuities.

The next proposition provides formulas on $\varrho_h(L)$ for monotone payoffs, from which we also show that $C\mapsto \varrho_h(C)$ is monotone w.r.t.~$\preceq$.
In the following, a set $\mathcal S\subseteq[0,1]^2$ is called increasing if, for every $(u_1,v_1),(u_2,v_2)\in \mathcal S$, either ``$u_1\le u_2$ and $v_1\le v_2$'' or ``$u_1\ge u_2$ and $v_1\ge v_2$'' holds.

\begin{Proposition}[Formulas and monotonicity of distortion risk measures]\label{prop:monotonicity:distortion}
Suppose that $L$ is monotone of the form $L=g(K)$. Then the following hold.
\begin{enumerate}[label=(\roman*)]
\item\label{prop:item:i} There exist a sequence $\mathbf{z}=(z_m)_{m\in\mathbb{N}_0}$, $z_m\in [0,\infty)$, $(A_m,B_m)\in \mathbb{R}\times \{-1,1\}$, $m\in \mathbb{N}$, and an increasing set $\mathcal S=\{(u_m,v_m):{m\in \mathbb{N}}\}$, $(u_m,v_m)\in [0,1]^2$, $u_1\ge u_2 \ge \cdots$, such that 
\begin{align}\label{eq:canonical:form:rho:h}
\varrho_h(C)=z_0+\sum_{m\in\mathbb{N}}z_m\, h(A_m+B_m C(u_m,v_m)),
\end{align}
where $\mathbf{z}$ and $\mathcal S$ are  determined by $(g,F,G)$, and
$(A_m,B_m)\in \mathbb{R}\times \{-1,1\}$, $m\in \mathbb{N}$, are determined by $\mathcal S$ given by 
\begin{align*}
(A_m,B_m)=\begin{cases}
(0,1),&\text{ if $g$ is increasing and $K=K_\wedge$},\\\displaystyle
(u_m+v_m,-1),&\text{ if $g$ is increasing and $K=K_\vee$},\\\displaystyle
(1,-1),&\text{ if $g$ is decreasing and $K=K_\wedge$},\\\displaystyle
(1-u_m-v_m,1),&\text{ if $g$ is decreasing and $K=K_\vee$}.\\
\end{cases}
\end{align*}
\item\label{prop:item:ii} If $L$    is increasing (decreasing) in~$K_{\wedge}$, then $\varrho_h$ is increasing (decreasing) w.r.t.~$\preceq$.
   \item\label{prop:item:iii}  If $L$ is increasing (decreasing) in~$K_{\vee}$, then $\varrho_h$ is decreasing (increasing) w.r.t.~$\preceq$.
   \item\label{prop:item:iv} Suppose that
   human lifespan is limited, that is, 
   there exist $\omega_X,\omega_Y >0$ such that $X\le \omega_X$ and $Y\le \omega_Y$ almost surely.
   Then~\ref{prop:item:i} holds with a finite sum, that is, there exist $\bar m\in \mathbb{N}$, $z_0,z_m\in [0,\infty)$, $(A_m,B_m)\in \mathbb{R}\times \{-1,1\}$ and $(u_m,v_m)\in [0,1]^2$ for  $m=1,\dots,\bar m$ such that $u_1\ge  \cdots \ge u_{\bar m}$, $v_1\ge  \cdots \ge v_{\bar m}$ and that 
   \begin{align}\label{eq:canonical:form:rho:h:finite}
    \varrho_h(C)=z_0+\sum_{m=1}^{\bar m}z_m\, h(A_m+B_m C(u_m,v_m)),
\end{align}
where $\bar m$, $\mathbf{z}$ and $\mathcal S$ are determined by $(g,F,G)$, and
$(A_m,B_m)\in \mathbb{R}\times \{-1,1\}$ are given in~\ref{prop:item:i}.
    \end{enumerate}
\end{Proposition}

Several remarks on Proposition~\ref{prop:monotonicity:distortion} are in order.
First, $\mathbf{z}$, $\bar m$ and $\mathcal S$ are computed directly from $(g,F,G)$; see the proof of this proposition in Appendix~\ref{sec:proofs}.
Second, the assumption of limited lifespan in Proposition~\ref{prop:monotonicity:distortion}~\ref{prop:item:iv}  is statistically validated in some studies, such as~\citet{einmahl2019limits}.
Finally, due to the monotonicity of $\varrho_h$, bounds on $\varrho_h(C)$ can be computed based on known bounds on copulas w.r.t. $\preceq$.
The most well-known but widest bounds are given by the \emph{Fr{\'e}chet-Hoeffding inequality}: $W\preceq C\preceq M$ for every $C\in \mathcal C$, where $W(u,v)=\max(0,u+v-1)$, $(u,v)\in[0,1]^2$, is called the \emph{counter-monotonic} copula, and $M(u,v)=\min(u,v)$, $(u,v)\in[0,1]^2$ , is called the \emph{comonotonic copula}; see~\citet{nelsen2006introduction} for details.
Note that this inequality ensures that $A_m+B_m C(u_m,v_m)\in[0,1]$ for every $m\in\mathbb{N}$.
Under the assumption of \emph{positive quadrant/orthant dependence (PQD/POD)}, the lower bound becomes $\Pi \preceq C$, where $\Pi(u,v)=uv$, $(u,v)\in[0,1]^2$, is the \emph{independence copula}.
Further improved bounds may also be applicable as in the following remark.

\begin{Remark}[Improved Fr\'echet bounds]\label{rem:tankov}
We can extend the domain of the map $C\mapsto \varrho_h(C)$ from $\mathcal C$ to $\mathcal Q$, the set of all \emph{quasi-copulas}~\citep{alsina1993characterization,genest1999characterization} by simply replacing $C$ with $Q$ in~\eqref{eq:canonical:form:rho:h}.
We also write $Q_1\preceq Q_2$ for $Q_1,Q_2\in\mathcal Q$ if $Q_1(u,v)\le Q_2(u,v)$ for all $(u,v)\in[0,1]^2$.
If there exist $\underline Q,\overline Q\in\mathcal Q$ such that $\underline Q\preceq C\preceq \overline Q$ for every $C\in\mathcal D$, then the crude bounds $\varrho_h(\underline Q)$ and $\varrho_h(\overline Q)$ are available.
The functions $\underline Q$ and $\overline Q$ are known when, for example, $\mathcal D=\mathcal C_\rho(r):=\{C\in\mathcal C: \rho(C)=r\}$, the set of all copulas with a known value $r\in[\rho(W),\rho(M)]$ of a dependence measure $\rho:\mathcal C\rightarrow \mathbb{R}$ such as the \emph{Spearman's rho} and \emph{Kendall's tau};~see~\citet{nelsen2001bounds,tankov2011improved}.
For this type of uncertainty set, risk bounds for some classes of aggregation functions are available in~\citet{kaas2009worst}.

Another relevant case is $\mathcal D=\mathcal C_{\mathcal S'}(Q):=\{C\in \mathcal C: C(u,v)=Q(u,v)\text{ for all }(u,v)\in \mathcal S'\}$,
     the collection of all copulas whose values are specified via a quasi-copula $Q\in \mathcal Q$ only on a compact subset $\mathcal S'$ of $[0,1]^2$.    
This set may be useful when, for example, values of the underlying copula are left unspecified on the upper tail region due to the scarcity of available data on young generations.
Then~\citet{tankov2011improved} derived functions $A^{\mathcal S',Q},B^{\mathcal S',Q}\in \mathcal Q$ such that $B^{\mathcal S',Q}\preceq C \preceq A^{\mathcal S',Q}$ for every $C \in \mathcal C_{\mathcal S'}(Q)$.
For our purposes, let $\mathcal S'\subseteq \mathcal S$. 
Then it is known that $B^{\mathcal S',Q} \in \mathcal C_{\mathcal S'}(Q)$.
Moreover, one can numerically check the attainability of $A^{\mathcal S',Q}$; see~\citet{mardani2010bounds} and~\citet[][Theorem~2.3]{sadooghi2013sharp}. 
\end{Remark}

\section{Bounds with a given level of uncertainty}\label{sec:uncertainty:level}

In this section, we study bounds on $\varrho_h(C)$ over a ball around some prespecified copula.
Throughout this section, we assume that human lifespan is limited, and thus $\varrho_h(C)$ admits the representation~\eqref{eq:canonical:form:rho:h:finite}.

Since $\varrho_h(C)$ depends only on the values of the underlying copula $C$ on $\mathcal S$, we measure the distance between two copulas over the finite set $\mathcal S$.
Let $\|\cdot\|$ be any norm on $\mathbb{R}^{\bar m}$.
Then we consider the following uncertainty set:
\begin{align}\label{eq:uncertainty:set}
    \mathcal C_{\mathcal S,\epsilon}(C^{\operatorname{ref}})=\left\{C\in \mathcal C: 
    \|
    (C(u_m,v_m)-C^{\operatorname{ref}}(u_m,v_m))_{m\in\{1,\dots,\bar m\}}
    \|
    \le \epsilon\right\},
\end{align}
where we call $C^{\operatorname{ref}}\in \mathcal C$ the \emph{reference copula} and $\epsilon>0$ the \emph{level of uncertainty}.

\begin{Remark}[Choice of $\epsilon$]\label{rem:choice:epsilon}
We provide some insight on the choice of $\epsilon$, which is an adapted version of the discussion in Remark 3.2 of~\citet{bernard2024robust} to our context.
For a subset of copulas $\mathcal D\subseteq \mathcal C$, define
\begin{align}\label{eq:bar:epsilon}
    \overline \epsilon(\mathcal D)=\sup\{\|
    (C(u_m,v_m)-C^{\operatorname{ref}}(u_m,v_m))_{m\in\{1,\dots,\bar m\}}
    \| :C \in \mathcal D \}.
\end{align} 
If there is a set of candidate copulas 
$\mathcal{D}$ in mind, then one could choose $\overline \epsilon(\mathcal D)$ as the level of uncertainty.
The quantity $\overline \epsilon (\mathcal D)$ can be calculated straightforwardly when, for example, $\mathcal D$ is a finite set.
If $\preceq$ is a total order on $\mathcal D^{\ast}=\mathcal D\cup\{C^{\operatorname{ref}}\}$ and there exist  maximal and minimal elements $\overline C$ and $\underline C$ on $\mathcal D^{\ast}$, then the supremum in $\overline \epsilon(\mathcal D)$ is attainable at either $\overline C$ or $\underline C$.
This situation may typically be the case when $\mathcal D^{\ast}$ is taken from a class of some parametric copulas.
We also remark that the uncertainty set~\eqref{eq:uncertainty:set} depends not only on $\epsilon$ but also on $\mathcal S$, which is determined by $g$, $F$ and $G$ as stated in Proposition~\ref{prop:monotonicity:distortion}.
Consequently, the scale and appropriate choice of the level of uncertainty may vary across contracts and marginal lifetime distributions.
In light of this dependency, the bound $\overline \epsilon(\mathcal C)$ is a useful quantity for gauging the scale of the level of uncertainty.
We will discuss how to compute $\overline \epsilon(\mathcal C)$ in Sections~\ref{sec:reformulations} and~\ref{sec:linear:program}.
\end{Remark}

\subsection{Reformulations in mathematical programming}\label{sec:reformulations}

The next result shows that the uncertainty set $\mathcal C_{\mathcal S,\epsilon}(C^{\operatorname{ref}})$ yields a convex feasible region on $\mathbb{R}^{\bar m}$.

\begin{Proposition}[Reformulation of $\overline \varrho_h(\mathcal C_{\mathcal S,\epsilon}(C^{\operatorname{ref}}))$ and $\underline \varrho_h(\mathcal C_{\mathcal S,\epsilon}(C^{\operatorname{ref}}))$]\label{prop:bounds:rho:h}
     Suppose that $L$ is monotone.
     Then
      $\overline \varrho_h(\mathcal C_{\mathcal S,\epsilon}(C^{\operatorname{ref}}))$ and $\underline \varrho_h(\mathcal C_{\mathcal S,\epsilon}(C^{\operatorname{ref}}))$ are given by:
     \begin{align}
        \label{eq:worst:varrho:ref} \overline \varrho_h(\mathcal C_{\mathcal S,\epsilon}(C^{\operatorname{ref}}))=\sup\left\{z_0 + \sum_{m=1}^{\bar m} z_m\,h(r_m):(r_1,\dots,r_{\bar m})\in \mathcal R(\epsilon)\right\},\\
        \label{eq:best:varrho:ref}
         \underline \varrho_h(\mathcal C_{\mathcal S,\epsilon}(C^{\operatorname{ref}}))=\inf\left\{z_0 + \sum_{m=1}^{\bar m }z_m\,h(r_m):(r_1,\dots,r_{\bar m})\in \mathcal R(\epsilon) \right\},
     \end{align}
     for a non-empty, compact and convex set  $\mathcal R(\epsilon)\subset [0,1]^{\bar m}$. 
     Moreover, $\mathcal R(\epsilon)$
     is  of the form:
     \begin{align*}
         \mathcal R(\epsilon)=\bigcap_{j=1}^{4\bar m - 1}\{\mathbf{r}\in \mathbb{R}^{\bar m}: K_j(\mathbf{r})\le 0\},
     \end{align*}
      for convex functions $K_j:\mathbb{R}^{\bar m}\rightarrow \mathbb{R}$, $j=1,\dots,4\bar m - 1$, where $K_j$, $j=1,\dots,4\bar m-2 $, are linear functions and $K_{4\bar m - 1}(\mathbf{r})= \|\mathbf{r}-\mathbf{r}^{\operatorname{ref}}\|-\epsilon$
      with $\mathbf{r}^{\operatorname{ref}}=(r_1^{\operatorname{ref}},\dots,r_{\bar m}^{\operatorname{ref}})^\top$ given by 
\begin{align*}
r_m^{\operatorname{ref}}=A_m + B_m C^{\operatorname{ref}}(u_m,v_m),\quad m=1,\dots,\bar m.
\end{align*}
\end{Proposition}

The quantity $\overline \epsilon(\mathcal C)$ introduced in Remark~\ref{rem:choice:epsilon} can be reformulated analogously.

\begin{Corollary}[Reformulation of $\overline \epsilon(\mathcal C)$]\label{cor:bar:epsilon}
The quantity $\overline \epsilon(\mathcal C)$ defined in~\eqref{eq:bar:epsilon} is given by
        \begin{align}
        \label{eq:bar:epsilon:ref} \overline \epsilon(\mathcal C)=\max\left\{\|\mathbf{r}-\mathbf{r}^{\operatorname{ref}}\|:\mathbf{r}\in \mathcal R \right\},
     \end{align}
     where $\mathbf{r}^{\operatorname{ref}}=(r_1^{\operatorname{ref}},\dots,r_{\bar m}^{\operatorname{ref}})^\top$ with $r_m^{\operatorname{ref}}= C^{\operatorname{ref}}(u_m,v_m)$, $m=1,\dots,\bar m$, and
$\mathcal R$ is a non-empty, compact and convex set  $\mathcal R\subset [0,1]^{\bar m}$ of the form:
     \begin{align*}
         \mathcal R=\bigcap_{j=1}^{4\bar m-2}\{\mathbf{r}\in \mathbb{R}^{\bar m}: K_j(\mathbf{r})\le 0\},
     \end{align*}
      for some linear functions $K_j:\mathbb{R}^{\bar m}\rightarrow \mathbb{R}$, $j=1,\dots,4\bar m-2$.
\end{Corollary}

Proposition~\ref{prop:bounds:rho:h} and Corollary~\ref{cor:bar:epsilon} reduce the problems of computing  the bounds $ \overline \varrho_h(\mathcal C_{\mathcal S,\epsilon}(C^{\operatorname{ref}}))$, $ \underline \varrho_h(\mathcal C_{\mathcal S,\epsilon}(C^{\operatorname{ref}}))$ and $\overline \epsilon(\mathcal C)$ to mathematical programming.
Note that, if $h$ is continuous, then the supremum in~\eqref{eq:worst:varrho:ref} and the infimum in~\eqref{eq:best:varrho:ref} are attainable by the extreme value theorem~\citep{Rudin1976}, 
and the maps
$\epsilon\mapsto \overline \varrho_h(\mathcal C_{\mathcal S,\epsilon}(C^{\operatorname{ref}}))$ and $\epsilon\mapsto \underline \varrho_h(\mathcal C_{\mathcal S,\epsilon}(C^{\operatorname{ref}}))$ are continuous by Berge's maximum theorem~\citep{Berge1963}.

Difficulties in computing these quantities depend on the choice of $h$ and the norm $\|\cdot\|$.
For example, convex maximization over a convex polytope, as in~\eqref{eq:worst:varrho:ref} with a convex $h$ and~\eqref{eq:bar:epsilon:ref}, is in general harder to solve than convex minimization as in~\eqref{eq:best:varrho:ref} with a convex $h$.
Nevertheless, in the next two subsections, we show that computing~\eqref{eq:worst:varrho:ref},~\eqref{eq:best:varrho:ref} and~\eqref{eq:bar:epsilon:ref} can be simplified for all distortion functions $h$ of our interest. 
In particular, when the norm is the $\mathcal L^1$-norm or the $\mathcal L^\infty$-norm, the problem is typically reduced to a combination of linear programs.

\subsection{Further reformulation for VaR and ES}\label{sec:reformulation:VaR:ES}

In this section, we show that the problems of computing~\eqref{eq:worst:varrho:ref} and~\eqref{eq:best:varrho:ref} can be further simplified for the cases of VaR and ES. In the following propositions, a sum $\sum_{m=n_1}^{n_2}z_m$, $n_1,n_2\in\mathbb{N}_0$ and $z_m\in \mathbb{R}$, is interpreted as $0$ if $n_1 > n_2$. 

First, we consider the case where $h$ is $h_\alpha^{\operatorname{VaR}}$.
In this case, $h$ has a jump at $1-\alpha$, and thus is neither convex nor concave.
Nevertheless,  the bounds $\underline{\operatorname{VaR}}_\alpha(\mathcal C_{\mathcal S,\epsilon}(C^{\operatorname{ref}}))$ and $\overline{\operatorname{VaR}}_\alpha(\mathcal C_{\mathcal S,\epsilon}(C^{\operatorname{ref}}))$ 
are attainable since~\eqref{eq:worst:varrho:ref} and~\eqref{eq:best:varrho:ref} reduce to the supremum and infimum over a subset of the finite set $\{z_0 + \sum_{m\in \mathcal I} z_m: \mathcal I\subseteq \{1,\dots,\bar m \}\}$, respectively.
Furthermore, the following proposition shows that each bound can be computed by solving convex programs at most $\bar m$ times.

\begin{Proposition}[VaR bounds on $\mathcal C_{\mathcal S,\epsilon}(C^{\operatorname{ref}})$]\label{prop:var:bounds}
Define the following quantities:
\begin{align*}
    \overline r_m &= \max\{r_m: \mathbf{r}\in\mathcal R(\epsilon)\},\\
    \underline r_m &= \min\{r_m: \mathbf{r}\in\mathcal R(\epsilon)\},\quad m\in\{1,\dots,\bar m\}.
    \end{align*}
    \begin{enumerate}[label=(\Roman*)]
\item\label{VaR:bounds:item:I} Suppose that $g$ is increasing. 
Define 
\begin{align*}
m_{\operatorname{L}}&=\min(\{m \in \{1,\dots,\bar m\}: \underline r_m \le 1-\alpha\}\cup\{\bar m + 1\}),\\
m_{\operatorname{U}}&=\max(\{m \in \{1,\dots,\bar m\}: \overline r_m > 1-\alpha\}\cup\{0\}).
   \end{align*}
              Then we have the following bounds:
\begin{align*}
    \underline{\operatorname{VaR}}_\alpha(\mathcal C_{\mathcal S,\epsilon}(C^{\operatorname{ref}}))
   & =\sum_{m=0}^{m_{\operatorname{L}}-1 }z_m, \\
    \overline{\operatorname{VaR}}_\alpha(\mathcal C_{\mathcal S,\epsilon}(C^{\operatorname{ref}}))&=
\sum_{m=0}^{m_{\operatorname{U}}}z_m.
\end{align*}

\item\label{VaR:bounds:item:II} Suppose that $g$ is decreasing. Define
    \begin{align*}
m_{\operatorname{L}}&=
\max(\{m \in \{1,\dots,\bar m\}: \underline r_m \le 1-\alpha\}\cup\{0\}),\\
m_{\operatorname{U}}&=
\min(\{m \in \{1,\dots,\bar m\}: \overline r_m > 1-\alpha\}\cup\{\bar m + 1\}).
   \end{align*}
              Then we have the following bounds:
\begin{align*}
\underline{\operatorname{VaR}}_\alpha(\mathcal C_{\mathcal S,\epsilon}(C^{\operatorname{ref}}))
    &=z_0+\sum_{m=m_{\operatorname{L}}+1}^{\bar m}z_m,\\ 
\overline{\operatorname{VaR}}_\alpha(\mathcal C_{\mathcal S,\epsilon}(C^{\operatorname{ref}}))&=
\displaystyle z_0 +\sum_{m=m_{\operatorname{U}}}^{\bar m}z_m.
\end{align*}
    \end{enumerate}

\end{Proposition}

By Proposition~\ref{prop:var:bounds}, calculation of each VaR bound reduces to finding a corresponding index.
This index can be found by solving convex programs at most $\bar m$ times since $\{\overline r_m\}_{m\in\{1,\dots,\bar m\}}$ and $\{\underline r_m\}_{m\in\{1,\dots,\bar m\}}$ are monotone sequences.
Regarding the sensitivity to~$\epsilon$, the bounds are generally not continuous w.r.t.~$\epsilon$.

Next, we consider the case of ES.
Since $h_\alpha^{\operatorname{ES}}$ is continuous, the supremum in~\eqref{eq:worst:varrho:ref} and the infimum in~\eqref{eq:best:varrho:ref} are both attainable.
Due to the piecewise linearity of $h_\alpha^{\operatorname{ES}}$, the bounds can be computed similarly to the case of VaR by the following proposition.

\begin{Proposition}[ES bounds on $\mathcal C_{\mathcal S,\epsilon}(C^{\operatorname{ref}})$]\label{prop:es:bounds}
\begin{enumerate}[label=(\Roman*)]
\hspace{0mm}
\item\label{ES:bounds:item:I} Suppose that $g$ is increasing. 
Define
\begin{align*}
    \mathcal R_m(\epsilon)=\{\mathbf{r}\in \mathcal R(\epsilon): r_{m+1} \le 1-\alpha \le  r_{m}\},\quad m=0,\dots,\bar m,
\end{align*}
where inequalities on $r_0$ and $r_{\bar m + 1}$ are interpreted to hold for any $\mathbf{r}\in\mathcal R(\epsilon)$.
For $m=0,\dots,\bar m$, define $\overline H_{m}$ and  $\underline H_{m}$ as the supremum and infimum of 
\begin{align}\label{eq:H:ES:I}
\mathbf{r}\mapsto \sum_{m'=0}^{m} z_{m'}  +  \sum_{m'=m+1}^{\bar m} \frac{z_{m'} }{1-\alpha}r_{m'}
\end{align}
over $\mathcal R_{m}(\epsilon)$, respectively.
\item\label{ES:bounds:item:II} Suppose that $g$ is decreasing. 
Define:
\begin{align*}
    \mathcal R_m(\epsilon)=\{\mathbf{r}\in \mathcal R(\epsilon): r_{m} \le 1-\alpha \le r_{m+1}\},\quad m=0,\dots,\bar m,
\end{align*}
where inequalities on $r_0$ and $r_{\bar m + 1}$ are interpreted to hold for any $\mathbf{r}\in\mathcal R(\epsilon)$.
For $m=0,\dots,\bar m$, define $\overline H_{m}$ and  $\underline H_{m}$ as the supremum and infimum of 
\begin{align}\label{eq:H:ES:II}
\mathbf{r}\mapsto 
z_0 + \sum_{m'=m+1}^{\bar m} z_{m'}  +  \sum_{m'=1}^{m} \frac{z_{m'} }{1-\alpha}r_{m'}
\end{align}
over $\mathcal R_{m}(\epsilon)$, respectively.
\end{enumerate}
In both cases,
\begin{align*}
\underline{\operatorname{ES}}_\alpha(\mathcal C_{\mathcal S,\epsilon}(C^{\operatorname{ref}}))&=
\min\{\underline H_{m}:m\in\{0,\dots,\bar m\}\},\\
\overline{\operatorname{ES}}_\alpha(\mathcal C_{\mathcal S,\epsilon}(C^{\operatorname{ref}}))&=
\max\{\overline H_{m}:m\in\{0,\dots,\bar m\}\},
\end{align*}
and all the bounds are attainable.
\end{Proposition}

Note that, for $m\in\{0,\dots,\bar m\}$ such that $\mathcal R_m(\epsilon)=\emptyset$, we have that $\underline H_m = \infty$ and $\overline H_m = -\infty$.
For $m\in\{0,\dots,\bar m\}$ with $\mathcal R_m(\epsilon)\neq \emptyset$, the problems for finding $\underline H_m$ and $\overline H_m$ are convex programs since $\mathcal R_m(\epsilon)$ is a closed convex set and the functions in~\eqref{eq:H:ES:I} and~\eqref{eq:H:ES:II} are linear.

\subsection{Computation with linear programming}\label{sec:linear:program}

In this section, we consider the case that the norm $\|\cdot\|$ in~\eqref{eq:uncertainty:set} and~\eqref{eq:bar:epsilon} is 
the $\mathcal L^1$-norm $\|\cdot\|_1$ or the $\mathcal L^\infty$-norm $\|\cdot\|_{\infty}$.
We show that, in this case, all the optimization problems in Section~\ref{sec:reformulation:VaR:ES} reduce to combinations of linear programs.
Here, we refer to a \emph{linear program (LP)} as the optimization of a linear objective function under linear constraints, which can be solved efficiently by standard software.
For reference, Appendix~\ref{sec:appendix:programs} provides more detailed descriptions of linear programs required to compute the risk bounds.

The next result follows directly from standard techniques as in~\citet[][Chapter~6]{boyd2004convex}. 

\begin{Lemma}[The case of  $\mathcal L^1$- or $\mathcal L^{\infty}$-norm]
Suppose that $L$ is monotone.
    \begin{enumerate}[label=(\roman*)]
        \label{lem:lp:norm}
        \item\label{lem:item:i} 
        If $\|\cdot\|=\|\cdot\|_1$, the set $\mathcal R(\epsilon)$ in~\eqref{eq:worst:varrho:ref} and~\eqref{eq:best:varrho:ref} can be replaced by
     \begin{align*}
         \mathcal R_{1}(\epsilon)=\bigcap_{j=1}^{6\bar m-1}\left\{
         \begin{pmatrix}
    \mathbf{r}\\             \mathbf{s}\\
\end{pmatrix}\in \mathbb{R}^{2\bar m}: K_j(\mathbf{r},\mathbf{s})\le 0\right\},
     \end{align*}
      for some linear functions $K_j:\mathbb{R}^{2\bar m}\rightarrow \mathbb{R}$, $j=1,\dots,6\bar m-1$, where $\mathbf{s}\in\mathbb{R}^{\bar m}$ is a vector of auxiliary variables. 
\item\label{lem:item:ii} 
        If $\|\cdot\|=\|\cdot\|_{\infty}$, then the set $\mathcal R(\epsilon)$ in~\eqref{eq:worst:varrho:ref} and~\eqref{eq:best:varrho:ref} can be replaced by
     \begin{align*}
         \mathcal R_{\infty}(\epsilon)=\bigcap_{j=1}^{6\bar m-2}\{\mathbf{r}\in \mathbb{R}^{\bar m}: K_j(\mathbf{r})\le 0\},
     \end{align*}
      for some linear functions $K_j:\mathbb{R}^{\bar m}\rightarrow \mathbb{R}$, $j=1,\dots,6\bar m-2$.
    \end{enumerate}
\end{Lemma}

Assume that the norm $\|\cdot\|$ is  $\|\cdot\|_1$ or $\|\cdot\|_{\infty}$.
Then the following are implied from Lemma~\ref{lem:lp:norm}~\ref{lem:item:i} and~\ref{lem:item:ii}.
First, when $h=h^{\mathbb{E}}$, computations of~\eqref{eq:worst:varrho:ref} and~\eqref{eq:best:varrho:ref} can be reformulated as LPs.
Second, calculations of $\overline r_m$ and $\underline r_m$ in Proposition~\ref{prop:var:bounds} are LPs, and thus each VaR bound can be computed by solving LPs at most $\bar m$ times.
Finally, calculations of $\overline H_m$ and $\underline H_m$ in Proposition~\ref{prop:es:bounds} are also LPs, and thus each ES bound can be computed by solving LPs at most $\bar m + 1$ times.

Next, we comment on the computation of~\eqref{eq:bar:epsilon:ref}.
For $\mathcal R$ in Corollary~\ref{cor:bar:epsilon}, define
$\overline r_m = \max\{r_m: \mathbf{r}\in\mathcal R\}$, $
    \underline r_m = \min\{r_m: \mathbf{r}\in\mathcal R\}$ and 
    $r_m^\ast = |\overline r_m - r_m^{\text{ref}}|\vee |\underline r_m - r_m^{\text{ref}}|$
    for $m\in\{1,\dots,\bar m\}$.
Then $\overline \epsilon(\mathcal C) = \max_{m\in\{1,\dots,\bar m\}} r_m^\ast$ if $\|\cdot\|=\|\cdot\|_{\infty}$.
If $\|\cdot\|=\|\cdot\|_{1}$, we have the upper bound $\overline \epsilon(\mathcal C)\le \sum_{m=1}^{\bar m} r_m^\ast=:\overline \epsilon^\ast(\mathcal C)$, and computing the exact value of $\overline \epsilon(\mathcal C)$ may require computationally intensive methods such as vertex enumeration.

We end this section by noting that the reduction to linear programming is not limited to the $\mathcal{L}^1$ and $\mathcal{L}^\infty$-norms but extends to any norms whose $\epsilon$-ball is a polyhedron, such as the \emph{CVaR norm} (also known as the D-norm or largest-$k$ norm) studied in~\citet{gotoh2016two}.

\section{Numerical experiments}\label{sec:numerical}

This section numerically compares the bounds derived in Section~\ref{sec:uncertainty:level} and those in Remark~\ref{rem:tankov}.
For marginal distributions of $X$ and $Y$, we consider the \emph{Gompertz law} with an upper truncation, whose cdf is given by
\begin{align*}
F(x;t,m,\sigma,\omega)=\frac{\tilde F(x;t,m,\sigma)}{\tilde F(\omega;t,m,\sigma)},\quad 0\le x\le \omega,
\end{align*}
where $t\ge 0, m,\sigma,\omega>0$ and $\tilde F$ is given by
\begin{align*}
        \tilde F(x\,;\, t,m,\sigma) = 1-\exp\left[
    \exp\left(\frac{t-m}{\sigma}\right)\cdot \left\{1-\exp\left(\frac{x}{\sigma}\right)\right\}
    \right], \quad x \ge 0.
\end{align*}
Note that $t$ is the age at the beginning of the contract, $m$ is the mode parameter, $\sigma$ is the dispersion parameter and $\omega$ is the maximum lifetime.

Let $(t_X, m_X,\sigma_X,\omega_X)$ and $(t_Y, m_Y,\sigma_Y,\omega_Y)$ be the sets of parameters for $X$ and $Y$, respectively.
We set $(m_X,\sigma_X)=(85.47,10.45)$ and $(m_Y,\sigma_Y)=(91.57,8.13)$ according to Table~5 of~\citet{dufresne2018age}; see also~\citet{frees1996annuity} and~\citet{carriere2000bivariate}.
We set $\omega_X=115 - t_X$ and  $\omega_Y=115 - t_Y$, with $t_X$ and $t_Y$ determined differently for contracts.
For the contracts, we consider the following four simple ones with $v=1/(1+0.05)$:
\begin{enumerate}[label=(\roman*)]
\item \emph{First-to-die annuity (F2DA)}: $(t_X,t_Y)=(35,32)$ and $L=\sum_{k=1}^{K_\wedge}a\,v^k$, $a=1$, which is increasing in $K_\wedge$;
\item \emph{Second-to-die annuity (S2DA)}: $(t_X,t_Y)=(65,62)$ and $L=\sum_{k=1}^{K_\vee}a\,v^k$,  $a= 1.169$, which is increasing in $K_\vee$;
\item \emph{First-to-die insurance (F2DI)}: $(t_X,t_Y)=(65,62)$ and $L=b\,v^{K_\wedge +1}$, $b=35.036$, which is decreasing in $K_\wedge$;
\item \emph{Second-to-die insurance (S2DI)}: $(t_X,t_Y)=(65,62)$ and $L=b\,v^{K_\vee+1}$, $b=63.531$, which is decreasing in $K_\vee$.
\end{enumerate}
The quantities $a,b\ge 0$ are taken such that the price $\mathbb{E}[L]$ is equal for all the contracts when the underlying survival copula is the independence copula.

\begin{figure}[t!]
    \centering
\includegraphics[width=0.8\textwidth]{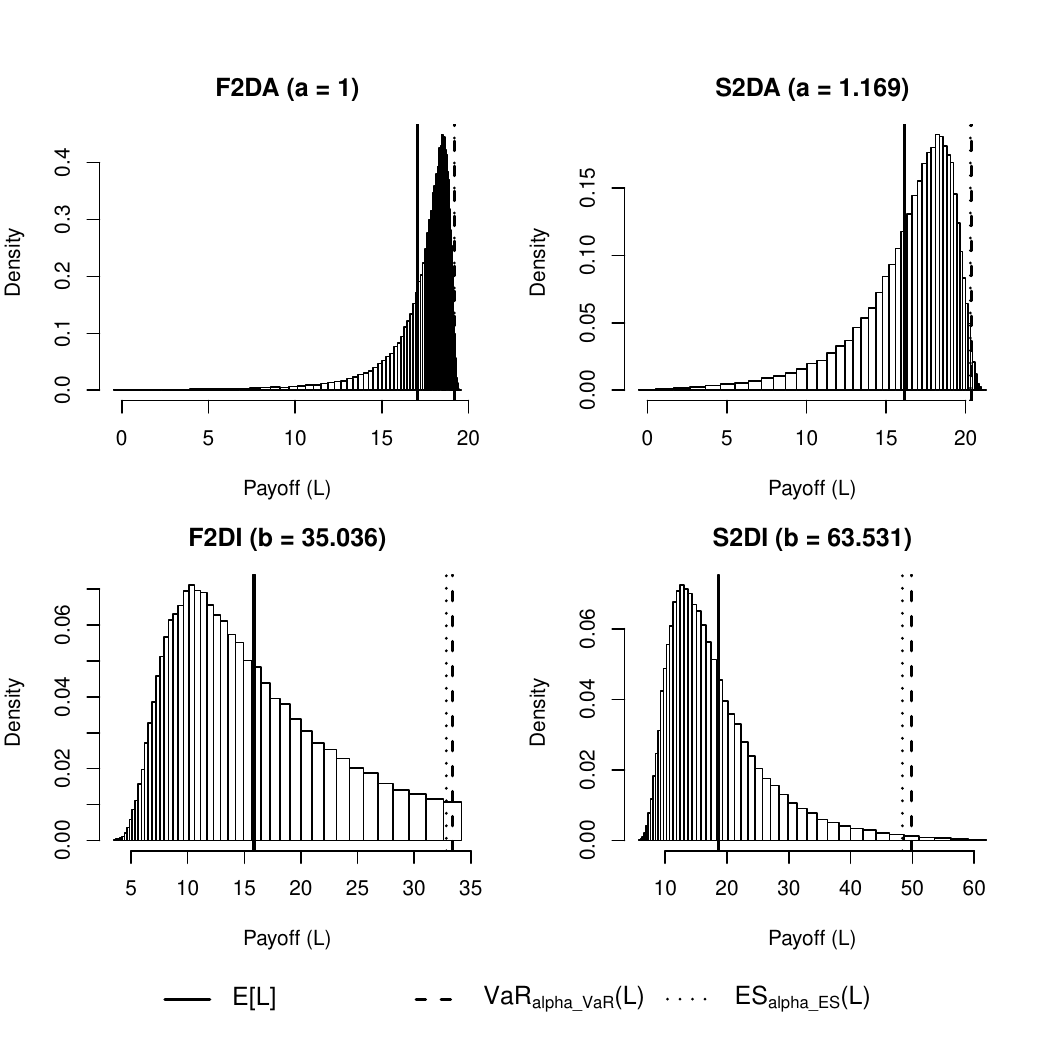}
    \caption{Histograms of simulated payoffs ($L$) for the four contracts (F2DA, S2DA, F2DI and S2DI) under the survival Gumbel copula with $\tau = 0.49$. Vertical lines indicate the true values for $\mathbb{E}[L]$ (solid), VaR at 99\% (dashed), and ES at 97.5\% (dotted).
    The number of Monte Carlo replications is $10^5$.}
    \label{fig:risk_histograms}
\end{figure}

For the reference copula $C^{\text{ref}}$ of $(X,Y)$, we take the survival Gumbel copula according to~\citet{dufresne2018age}.
The Gumbel copula is 
$C_\delta(u, v) = \exp \left[- \left\{ (-\ln u)^{\delta} +  (-\ln v)^{\delta} \right\}^{1/\delta} \right]$, 
$u,v\in(0,1)$, where $\delta\in[1,\infty)$ is the dependence parameter.
Note that this copula corresponds to $\Pi$ and $M$ as $\delta \rightarrow 1$ and $\delta \rightarrow \infty$, respectively.
Moreover, $C_\delta$ has the \emph{Kendall's tau} $\tau = (\delta-1)/\delta$,
lower \emph{tail order}~\citep{hua2011tail} $\kappa_{\operatorname{L}}=2^{1/\delta}$
and the upper \emph{tail dependence coefficient} $\lambda_{\operatorname{U}} =2-2^{1/\delta}$; see~\citet{joe2014dependence}.
According to Table~6 of~\citet{dufresne2018age}, we choose $\delta=1.96$ for this reference copula, which leads to  $\tau =0.49$, $\kappa_{\operatorname{L}}=2^{1/\delta}= 1.42$
and $\lambda_{\operatorname{U}} =2-2^{1/\delta}=0.58$.
Based on the standard error in this table, the $3$-sigma range for $\delta$ is $(1.90, 2.02)$.

For VaR and ES, we consider the standard confidence levels $\alpha_{\operatorname{VaR}}=0.99$ and $\alpha_{\operatorname{ES}}=0.975$, respectively.

Fig.~\ref{fig:risk_histograms} shows histograms of simulated payoffs for the four contracts above when $C=C^{\text{ref}}$.
The number of Monte Carlo replications is $10^5$.
For reference, Table~\ref{table:price:bounds} reports the prices of contracts analytically computed for various copulas, where the prices are computed by $\mathbb{E}[L]+\lambda \, \varrho_h(L)$ with $h$ being $h_{\alpha^{\operatorname{VaR}}}^{\operatorname{VaR}}$ or $h_{\alpha^{\operatorname{ES}}}^{\operatorname{ES}}$.
The loading parameter is $\lambda=0.06$.

    \begin{table}[t]
\caption{Comparison of insurance premiums $\mathbb{E}[L] + \lambda \VaR_{0.99}(L)$ and $\mathbb{E}[L] + \lambda \ES_{0.975}(L)$ for various copulas.
The loading parameter is $\lambda=0.06$.} \label{table:price:bounds}
\centering
\begin{tabular}{lllllllll}
  \toprule
 & \multicolumn{4}{c}{VaR loading} & \multicolumn{4}{c}{ES loading} \\
\cmidrule(lr){2-5} \cmidrule(lr){6-9}
 & M & SG & Indep & W & M & SG & Indep & W \\
\midrule
 F2DA & 18.3053 & 18.2014 & 18.0041 & 17.7671 & 18.3049 & 18.2008 & 18.0039 & 17.7671 \\ 
  S2DA & 17.0538 & 17.3838 & 18.0862 & 18.9004 & 17.0490 & 17.3791 & 18.0821 & 18.8962 \\ 
  F2DI & 17.3908 & 17.8617 & 18.8642 & 20.0262 & 17.3490 & 17.8301 & 18.8355 & 19.9976 \\ 
  S2DI & 22.8272 & 21.6670 & 19.2026 & 15.9966 & 22.7490 & 21.5827 & 19.2362 & 15.9858 \\ 
   \bottomrule
\end{tabular}
\end{table}

Next, we compare the following bounds on $\varrho_h(C)$ for $h=h^{\mathbb{E}}$, $h_{\alpha^{\operatorname{VaR}}}^{\operatorname{VaR}}$ and $h_{\alpha^{\operatorname{ES}}}^{\operatorname{ES}}$.
\begin{enumerate}[label=(\roman*)]\setlength{\itemsep}{2mm}
    \item \textbf{$\epsilon$-ball bounds}:      $\overline \varrho_h(\mathcal C_{\mathcal S,\epsilon}(C^{\operatorname{ref}}))$ and $\underline \varrho_h(\mathcal C_{\mathcal S,\epsilon}(C^{\operatorname{ref}}))$ presented in Section~\ref{sec:uncertainty:level} for different $\epsilon>0$.
    For the case of $\mathcal L^1$-norm, we compute the bounds for $\epsilon\in [0, \overline \epsilon^\ast(\mathcal C)]$; see Section~\ref{sec:linear:program}.
For the case of $\mathcal L^\infty$-norm, we compute them for $\epsilon\in [0, \overline \epsilon(\mathcal C)]$.
\item \textbf{Fr{\'e}chet-Hoeffding (FH) bounds}: $\varrho_h(W)$ and $\varrho_h(M)$. We also compute $\varrho_h(\Pi)$.
\item \textbf{Improved FH bounds with given values on a region}: 
 $\varrho_h(B^{\mathcal S',Q})$ and $\varrho_h(A^{\mathcal S',Q})$ in Remark~\ref{rem:tankov}, where $\mathcal S'=\{(u_m, v_m)\in \mathcal S: 0.2 \le u_m,v_m \le 0.8\}$ and
$Q=C^{\text{ref}}$.
\item 
\textbf{Improved FH bounds with a given value of Kendall's tau}: we also compute the bounds on the set of copulas with Kendall's tau given by $\tau = 0.49$; see Remark~\ref{rem:tankov}.
\end{enumerate}
The results are summarized in Fig.~\ref{fig:bounds:L1} for $\mathcal L^1$-norm and in  Fig.~\ref{fig:bounds:Linf} for $\mathcal L^\infty$-norm.

\begin{figure}[t]
    \centering
    \includegraphics[width=\textwidth]{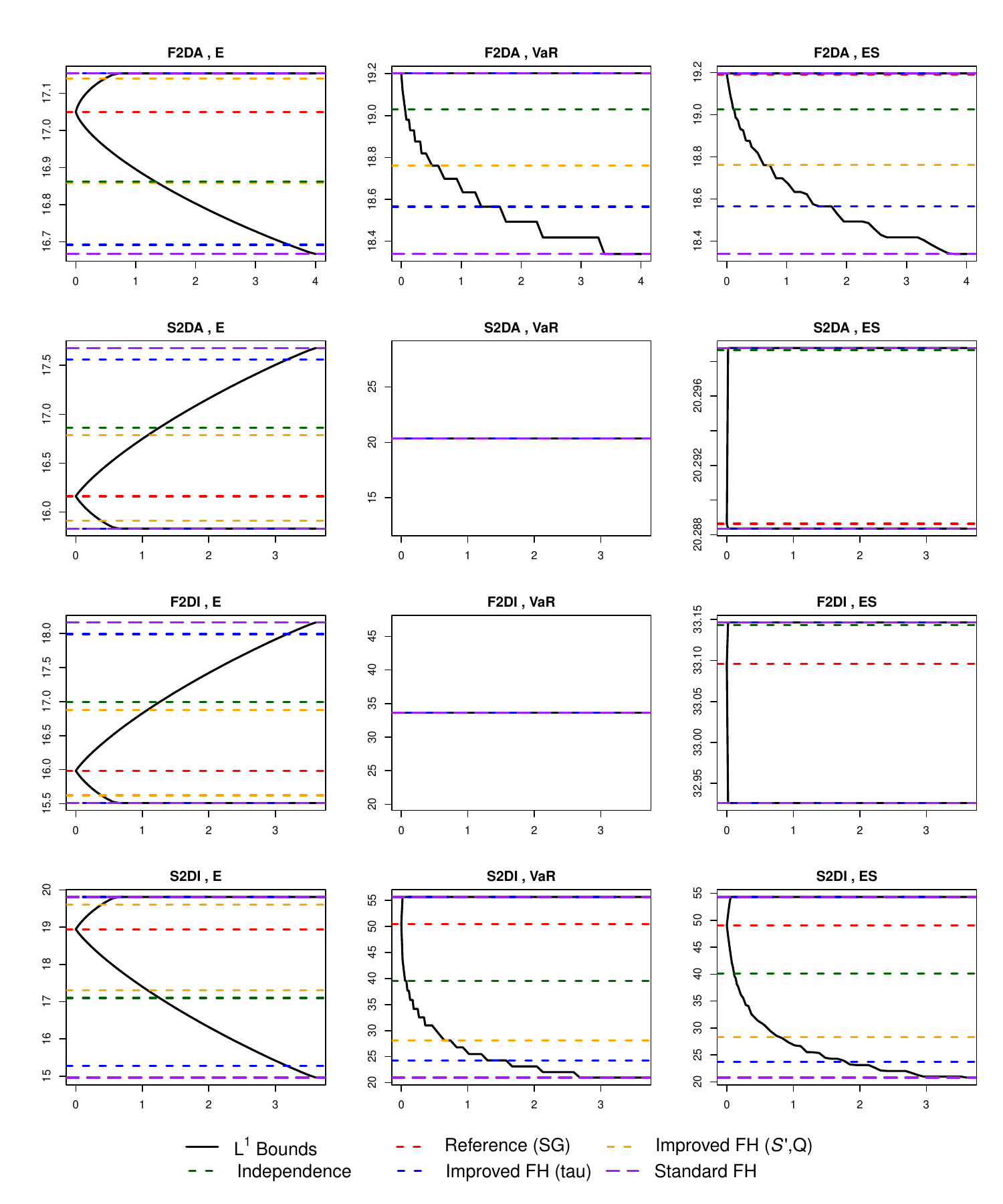}
    \caption{
     Bounds on the Expectation (E), 99\% VaR, and 97.5\% ES for the four contract types, plotted as a function of the $\mathcal L^1$ distance $\epsilon$ from a reference survival Gumbel copula with $\tau = 0.49$ (denoted by $C^{\text{ref}}$). The solid black lines represent the  bounds for a given $\epsilon$. The horizontal dashed lines show the values obtained under $C^{\text{ref}}$ (red), the independence copula (dark green), the improved Fr{\'e}chet-Hoeffding (FH) bounds given Kendall’s tau $\tau = 0.49$ (blue), the improved FH bounds given values $Q=C^{\text{ref}}$ on a subset $\mathcal S'=\{(u_m, v_m)\in \mathcal S: 0.2 \le u_m,v_m \le 0.8\}$ (orange), and the standard FH bounds (purple).
     }
    \label{fig:bounds:L1}
\end{figure}

\begin{figure}[t]
    \centering
    \includegraphics[width=\textwidth]{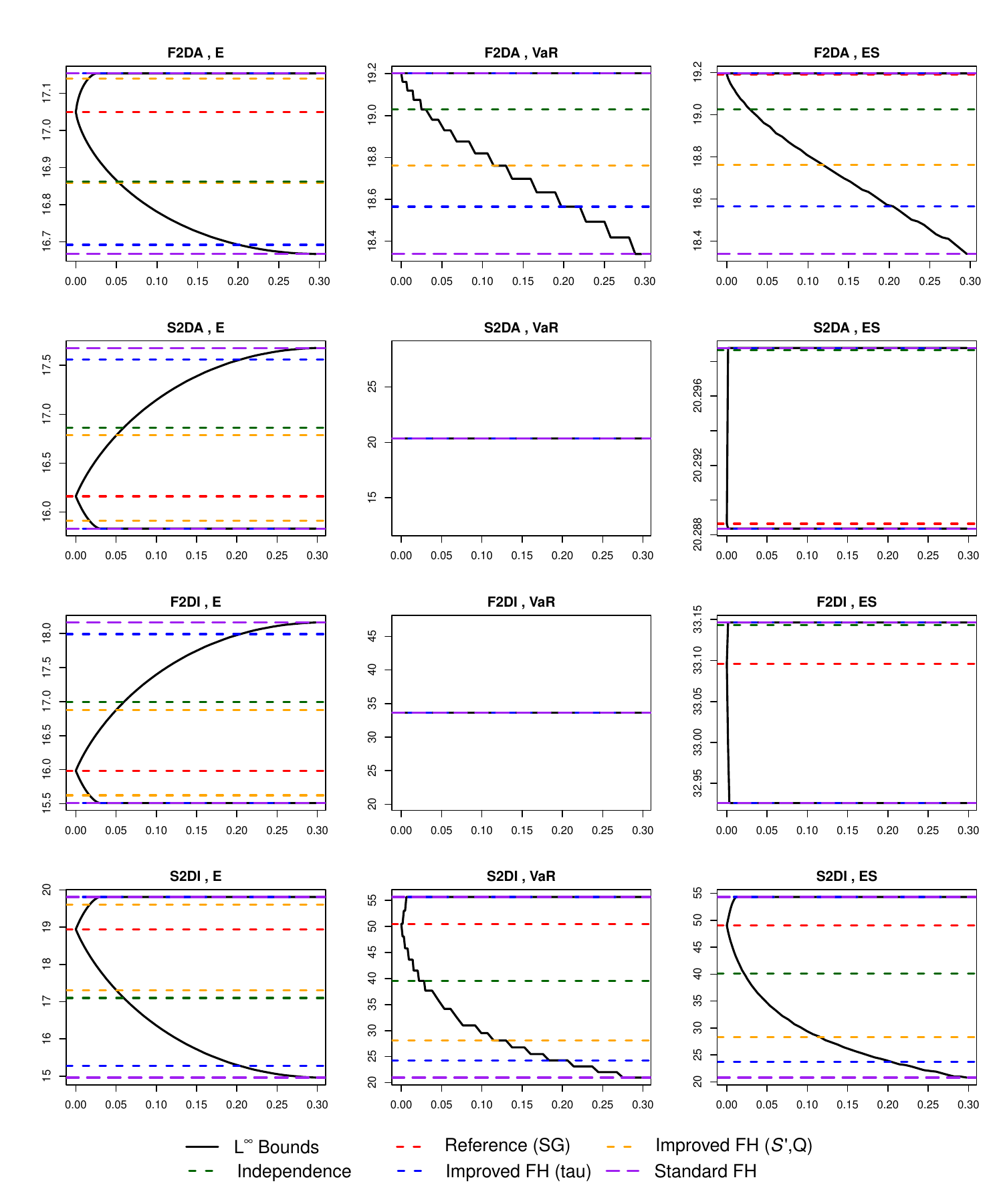}
    \caption{ Bounds on the Expectation (E), 99\% VaR, and 97.5\% ES for the four contract types, plotted as a function of the $\mathcal L^\infty$ distance $\epsilon$ from a reference survival Gumbel copula with $\tau = 0.49$. For the (quasi-) copulas in comparison, see the caption of Fig.~\ref{fig:bounds:L1}.}
    \label{fig:bounds:Linf}
\end{figure}

\begin{figure}[h!]
    \centering
    \includegraphics[width=\textwidth]{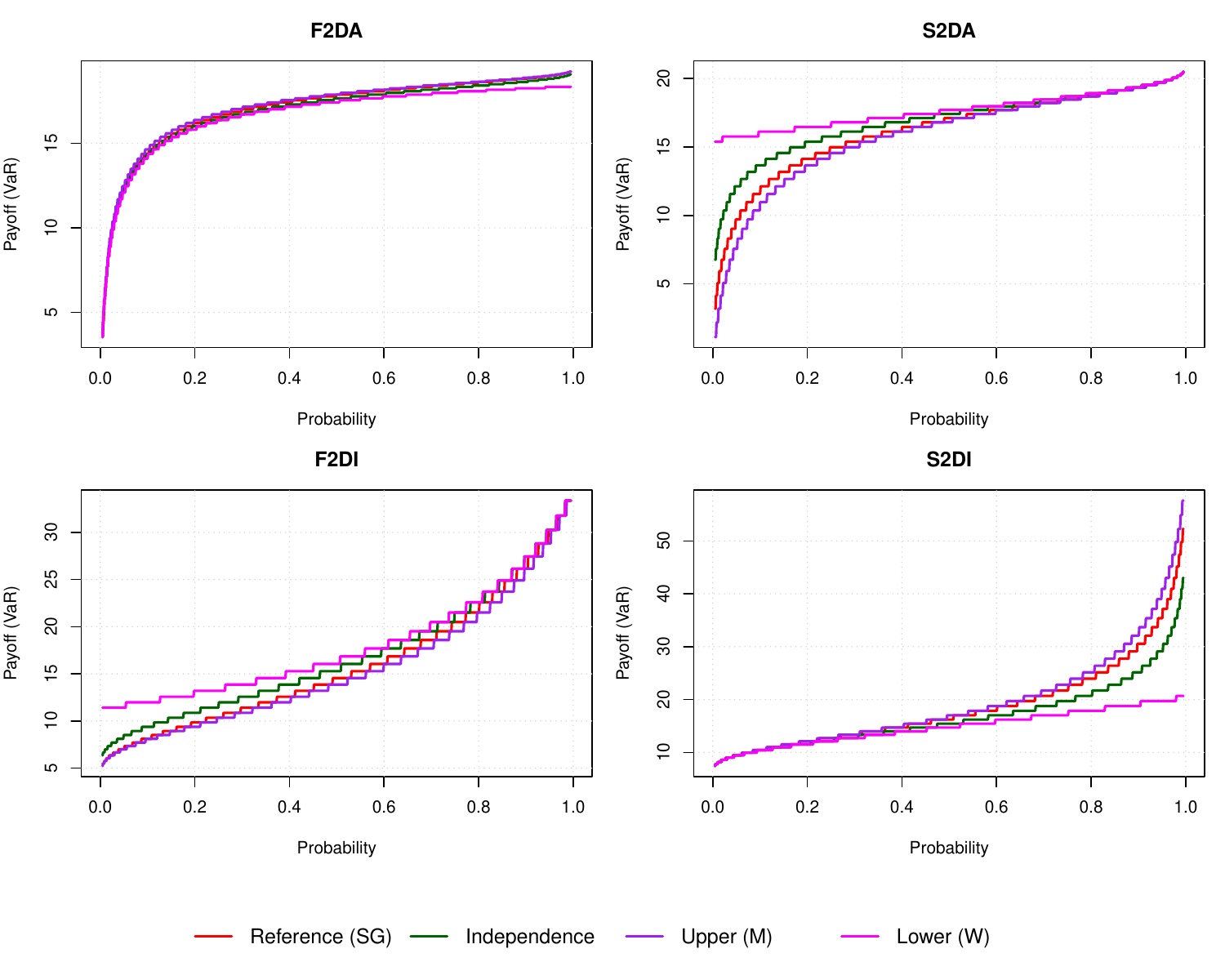}
    \caption{Empirical quantile curves of the payoffs for various copulas in comparison. The quantiles are computed based on simulated samples with the number of Monte Carlo replications $10^5$.
    }
    \label{fig:quantiles}
\end{figure}

We observe that the curves $\epsilon \mapsto \overline \varrho_h(\mathcal C_{\mathcal S,\epsilon}(C^{\operatorname{ref}}))$ and $\epsilon \mapsto \underline \varrho_h(\mathcal C_{\mathcal S,\epsilon}(C^{\operatorname{ref}}))$ behave similarly for the two cases of $\mathcal L^1$ and $\mathcal L^\infty$-norms.
Therefore, we focus on the case of $\mathcal L^1$-norm in Fig.~\ref{fig:bounds:L1}.

First, we comment on the horizontal lines in Fig.~\ref{fig:bounds:L1}.
Note that the improved FH bounds with $(\mathcal S',Q)$ share the reference copula $C^{\text{ref}}$ on the body part $\mathcal S'$.
Due to this construction, these improved FH bounds typically yield values closer to $\varrho_h(C^{\text{ref}})$ compared with other bounds especially for the case of expectation (first column in Fig.~\ref{fig:bounds:L1}).
Although the unspecified tail part $\mathcal S\backslash \mathcal S'$ is more relevant for the cases of VaR and ES, we observe that the improved FH bounds with $(\mathcal S',Q)$ yield tighter bounds than the improved FH bounds with given Kendall's tau. Whether independence underestimates or overestimates the risk depends on the type of the contract. For some contracts, the standard FH bounds yield values close to $\varrho_h(C^{\text{ref}})$, but the bounds under $W$ and $M$ may be too wide to use in practice.

Next, we discuss the behavior of the curves $\epsilon \mapsto \overline \varrho_h(\mathcal C_{\mathcal S,\epsilon}(C^{\operatorname{ref}}))$ and $\epsilon \mapsto \underline \varrho_h(\mathcal C_{\mathcal S,\epsilon}(C^{\operatorname{ref}}))$.
For the case of expectation, the upper and lower bounds continuously approach the FH bounds as $\epsilon$ increases.
Although this behavior is common for all the contracts, this is not the case for VaR and ES.
For F2DA, the copula $C^{\text{ref}}$ almost attains the upper bounds on VaR and ES.
For S2DI, we observe that $\overline \varrho_h(\mathcal C_{\mathcal S,\epsilon}(C^{\operatorname{ref}}))$ attains the FH bound with a tiny increase of $\epsilon$.
This indicates that S2DI may be highly sensitive to dependence uncertainty, and it may be advisable to evaluate VaR and ES by the worst cases if the reference copula is chosen without strong confidence.
For F2DA and S2DI, increasing the level of uncertainty mainly contributes to widening the lower risk bound.

    For S2DA and F2DI, we observe that VaR and ES are insensitive to dependence uncertainty.
To see this, Fig.~\ref{fig:quantiles} displays the quantile curves of the contracts for various copulas compared in this experiment.
The quantiles are computed based on simulated samples with the number of Monte Carlo replications $10^5$.
We observe that S2DA and F2DI have much tighter risk bounds in the upper tail regions than F2DA and S2DI. 
The difference in sensitivity among the contracts can be explained by considering tail events corresponding to extremely large payoffs.
For F2DA and S2DA, relevant tail events are  $\{T_{\wedge}\ge t\}$ and $\{T_{\vee}\ge t\}$, respectively, for $t$ close to the upper endpoint.
For F2DI and S2DI, risky scenarios are $\{T_{\wedge}\le t\}$ and $\{T_{\vee}\le t\}$, respectively, for $t$ close to the lower endpoint.
For F2DA and S2DI, the tail events are intersections of marginal tail events, whose probabilities are largely determined by the dependence structure.
On the other hand, for S2DA and F2DI, the corresponding events are unions of marginal tail events, whose probabilities are dominated by the sum of marginal probabilities, and thus the influence of the dependence structure can be relatively negligible.

Except for these special cases, our proposed $\epsilon$-ball bounds provide a spectrum of the interval between $\varrho_h(W)$ and $\varrho_h(M)$ via the uncertainty level $\epsilon$.
This spectrum of bounds may be useful for tightening the bounds, with robustness against the model specification taken into account, based on the current information on the dependence between $X$ and $Y$.
For example, when $\mathcal D$ is the set of survival Gumbel copulas where $\delta $ is in  the $3$-sigma range $(1.90,2.02)$ above, the resulting $\overline \epsilon(\mathcal D)$ for F2DA is approximately 0.05 for $\mathcal L^1$-norm and 0.002 for $\mathcal L^\infty$-norm. 
Based on Fig.~\ref{fig:bounds:L1} and Fig.~\ref{fig:bounds:Linf}, these values of $\epsilon$ offer much tighter bounds on $\varrho_h(C)$ than the existing ones.

\section{Concluding remarks}\label{sec:conclusion}

This paper proposes a novel framework for the robust risk evaluation of joint life insurance products in the presence of dependence uncertainty. 
The framework utilizes distortion risk measures to quantify the risk associated with the unknown dependence structure between the lifetimes of the two insureds. As the main contribution of this research, we show that the bounds for the expectation, VaR and ES can be computed by combinations of linear programs  when the uncertainty set is defined by an $\mathcal L^1$ or $\mathcal L^\infty$-ball centered around a reference copula.
Our numerical experiments demonstrate the benefits of the proposed framework, such as illustrating the sensitivity of the risk bounds against the level of dependence uncertainty and improving the existing risk bounds based on the available information.

It is of great interest to generalize the results in this paper in various directions.
The first direction is to extend the framework for broader classes of contracts, such as those with non-monotone payoffs and those involving more than two insureds. 
Second, it may be worth exploring to derive bounds based on other information, such as multivariate ageing properties~\citep{lai2006stochastic} and degrees of non-exchangeable (tail) dependence~\citep{koike2023measuring}.
Finally, analyzing the bounds for a portfolio of joint life insurance contracts can also be an interesting avenue for future research.

\backmatter

\bmhead{Supplementary information}
The \textsf{R} scripts to reproduce all numerical results are available upon request.

\bmhead{Acknowledgements}
The author is grateful to Haruki Tsunekawa at Nagoya University, Jun-ya Gotoh at Chuo University, Mario V. W\"{u}thrich at ETH Z\"urich and Silvana Pesenti at University of Toronto for their insightful comments on the early version of this paper.
Takaaki Koike is supported by JSPS KAKENHI Grant Number JP24K00273.

\section*{Declarations}
Not applicable.




\begin{appendices}

\section{Joint life insurance contracts}\label{sec:contracts}

This section presents insurance and pension contracts, whose prices are monotone w.r.t.~the concordance order $\preceq$ on $\mathcal C$.
For $k=1,2,\dots$, let $i_k>0$ be an effective annual rate of interest for the $k$th year, and let $v_k=1/(1+i_k)$ be the one-period discount factor.
We write $D_0=1$ and $D_k=v_1\cdots v_k$, $k\ge 1$, for the cumulative discount factor from time $k$ to time $0$. 
We assume that $v_1,v_2,\dots,$ are all non-random.
In the following, $n\in \mathbb{N}$ denotes the policy term.
For each $k\in\mathbb N_0$, define
\begin{align*}
p_k(C)&=\P(K_{\wedge}\ge k)=C(\bar F(k),\bar G(k)),\\
q_k(C)&=\P(K_{\vee}\ge k)=\bar F(k)+\bar G(k)-C(\bar F(k),\bar G(k)),
\end{align*}
for convenience.
Then $C\mapsto p_k(C)$ is increasing and $C\mapsto q_k(C)$ is decreasing  w.r.t.~$\preceq$.

\begin{enumerate}[label=(\roman*)]
    \item\label{item:joint:life:annuity} \emph{Joint life annuity}: suppose that an insurer pays $a_k\ge 0$ at the end of the $k$th year for $k\in\{1,\dots,n\}$ if both insureds are alive at that time.
    Then the random present value of the expense of this contract for the insurer is 
    \begin{align*}
    L=\sum_{k=1}^{n \wedge K_{\wedge}}a_k\, D_k,
    \end{align*}
    and its price is given by
    \begin{align*}
\mu(C)=        \sum_{k=1}^n a_k\, D_k \,\mathbb{P}(K_{\wedge}\ge k)=
\sum_{k=1}^n a_k\, D_k\,
p_k(C).
    \end{align*}
    Note that the \emph{joint life pure endowment} is the special case when $a_n>0$ and $a_k=0$ for $k\neq n$.
    \item \label{item:last:survivor:annuity} \emph{Last survivor annuity}: suppose that an insurer pays $a_k\ge 0$ at the end of the $k$th year for $k \in \{1,\dots,n\}$ if at least one of the two insureds is alive at that time.
    Then the payoff of this contract is 
\begin{align*}
    L=\sum_{k=1}^{n\wedge K_{\vee}}a_k\, D_k,
    \end{align*}
    and its price is given by
    \begin{align*}
\nonumber \mu(C)&=        \sum_{k=1}^n a_k\, D_k \,\mathbb{P}(K_{\vee}\ge k)\\
&=
\sum_{k=1}^n a_k\, D_k \,q_k(C)\\
&=
\sum_{k=1}^n a_k\, D_k\,\bar F(k)
+\sum_{k=1}^n a_k\, D_k\,\bar G(k)
-\sum_{k=1}^n a_k\, D_k\,
        C\left(
    \bar F(k),\bar G(k)
    \right).
    \end{align*}
The  \emph{last survivor pure endowment} is the special case when $a_n>0$ and $a_k=0$ for $k\neq n$.
    \item\label{item:joint:life:insurance} \emph{Joint life insurance}: suppose that an insurer pays $b_k\ge 0$ if the first death of the two insureds occurs in the $k$th year for $k\in \{1,\dots,n\}$.
    Then the payoff is given by
    \begin{align*}
    L=
    \sum_{k=0}^{n-1} b_{k+1} \, D_{k+1} \,\id_{\{K_{\wedge}=k\}}
    \end{align*}
    and its expectation is given by
  \begin{align*}
\nonumber \mu(C)&=   
\sum_{k=1}^n b_k \, D_k \,\mathbb{P}(k-1\le T_{\wedge}<k)\\&=
\sum_{k=1}^n b_k\, D_k\,
        \left \{ 
p_{k-1}(C)-p_k(C)
    \right\}\\
    &=b_1D_1 - b_n D_n p_n(C) + \sum_{k=1}^{n-1} \left(b_{k+1} D_{k+1}-b_{k}D_k \right)p_k(C).
    \end{align*}
    Therefore, if $b_k D_k \ge b_{k+1}D_{k+1}$ for $k=1,\dots,n-1$, then $C\mapsto \mu(C)$ is decreasing w.r.t.~$\preceq$.
    For the whole-life version with $L=b_{K_{\wedge}+1}D_{K_{\wedge}+1}$, the payoff is decreasing in $K_{\wedge}$ if $(b_kD_k)_{k\ge 1}$ is decreasing.
    \item\label{item:last:survivor:insurance} \emph{Last survivor insurance}: 
    suppose that an insurer pays $b_k\ge 0$ if the second death of the two insureds occurs in the $k$th year for $k\in \{1,\dots,n\}$.
    Then the payoff is given by
    \begin{align*}
    L=\sum_{k=0}^{n-1}b_{k+1}D_{k+1}\id_{\{K_{\vee}=k\}}.
    \end{align*}
Hence
    \begin{align*}
\nonumber \mu(C)&=   
\sum_{k=1}^n b_k \, D_k \,\mathbb{P}(k-1\le T_{\vee}<k)\\&=
\sum_{k=1}^n b_k\, D_k\,
        \left \{ 
q_{k-1}(C)-q_k(C)
    \right\}\\
    &=b_1D_1 - b_n D_n q_n(C) + \sum_{k=1}^{n-1} \left(b_{k+1} D_{k+1}-b_{k}D_k \right)q_k(C).
    \end{align*}
    Therefore, if $b_k D_k \ge b_{k+1}D_{k+1}$ for $k=1,\dots,n-1$, then $C\mapsto \mu(C)$ is increasing w.r.t.~$\preceq$.
        For the whole-life version with $L=b_{K_{\vee}+1}D_{K_{\vee}+1}$, the payoff is decreasing in $K_{\vee}$ if $(b_kD_k)_{k\ge 1}$ is decreasing.
    
    \end{enumerate}

The following contracts do not admit monotone payoffs, but their prices are monotone w.r.t.~$\preceq$.

\begin{enumerate}[label=(\roman*)]\setcounter{enumi}{4}
     \item\label{item:reversionary:annuity} \emph{Reversionary annuity}: suppose that an insurer pays $a_k\ge 0$ at the end of the $k$th year for $k\in\{1,\dots,n\}$ if only one insured is alive at that time.
     Then the payoff is given by
    \begin{align*}
    L=\sum_{k=1}^{n\wedge K_{\vee}}a_k\, D_k-\sum_{k=1}^{n\wedge K_{\wedge}}a_k\, D_k,
    \end{align*}
    and thus
    \begin{align*}
\nonumber \mu(C)&=       \sum_{k=1}^n a_k\, D_k \,\mathbb{P}(K_{\vee}\ge k)-   \sum_{k=1}^n a_k\, D_k \,\mathbb{P}(K_{\wedge}\ge k)\\
&=\sum_{k=1}^n a_k\, D_k\,\bar F(k)
+\sum_{k=1}^n a_k\, D_k\,\bar G(k)
-2\sum_{k=1}^n a_k\, D_k\,
        C\left(
    \bar F(k),\bar G(k)
    \right),
    \end{align*}
    which is decreasing w.r.t.~$\preceq$.
         \item\label{item:widow:pension} \emph{Widow's pension}: suppose that an insurer pays $a_k\ge 0$ at the end of $k$th year for $k\in \{1,\dots,n\}$ if the second insured corresponding to $Y$ is not alive and the first one associated with $X$ is alive at that time.
     Then the payoff is given by
    \begin{align*}
    L=\sum_{k=1}^{n\wedge K_{X}}a_k\, D_k-\sum_{k=1}^{n\wedge K_{\wedge}}a_k\, D_k,
    \end{align*}
    and thus
    \begin{align*}
\mu(C)&=\sum_{k=1}^n a_k\, D_k\,\bar F(k)
 -\sum_{k=1}^n a_k\, D_k\,
    C\left(
    \bar F(k),\bar G(k)
    \right),
    \end{align*}
        which is decreasing w.r.t.~$\preceq$.
\end{enumerate}

\section{Proofs}\label{sec:proofs}

\begin{proof}[Proof of Proposition~\ref{prop:monotonicity:distortion}]

For $t\ge 0$, write
\begin{align*}
    p_t(C)&= C\left(\bar F(t),\bar G(t)\right),\\
    q_t(C)&=\bar F(t)+\bar G(t)-C\left(\bar F(t),\bar G(t)\right),
\end{align*}
where $C$, the survival copula of $(X,Y)$, is the cdf of $(\bar F(X),\bar G(Y))$.
For every $t \in \mathbb{N}_0=\{0,1,\dots\}$, we have that
\begin{align*}
    \mathbb{P}(K_{\wedge}\ge t)&=\mathbb{P}(T_{\wedge}\ge t)=p_t(C),\\
    \mathbb{P}(K_{\vee}\ge t)&=\mathbb{P}(T_{\vee}\ge t)=q_t(C).
\end{align*}
Recall also that $h(0)=0$ and $h(1)=1$.

We first consider the case where $L=g(K_{\wedge})$ with $g$ being an increasing function.
Let
\begin{align*}
    \ell_k&=g(k)\ge 0,\quad k\in \mathbb{N}_0,\\
    \ell_{\infty}&=\lim_{k\to \infty}g(k)\in [0,\infty].
\end{align*}
If $x<\ell_0$, then $\mathbb{P}(L\le x)=0$.
If $\ell_{\infty}<\infty$ and $x\ge \ell_{\infty}$, then $\mathbb{P}(L>x)=0$.
For $\ell_0\le x<\ell_{\infty}$, let
\begin{align*}
    \bar k=\bar k(x)=\max\{k \in \mathbb{N}_0:\ell_k \le x\}.
\end{align*}
Then
\begin{align*}
    \mathbb{P}(L\le x)
    &=\sum_{k\in\mathbb{N}_0}\mathbb{P}(L\le x\mid K_{\wedge}=k)\mathbb{P}(K_{\wedge}=k)\\
    &=\sum_{k\in\mathbb{N}_0}\id_{\{\ell_k\le x\}}\{p_k(C)-p_{k+1}(C)\}\\
    &=\sum_{k=0}^{\bar k(x)}\{p_k(C)-p_{k+1}(C)\}\\
    &=1-p_{\bar k(x)+1}(C).
\end{align*}
Therefore,
\begin{align*}
    \varrho_h(C)
    =\int_0^\infty h(\mathbb{P}(L>x))\,\rd x
    =\ell_0+\int_{[\ell_0,\ell_\infty)} h\left(p_{\bar k(x)+1}(C)\right)\,\rd x.
\end{align*}
Recursively define the increasing sequence $k_0,k_1,\dots \in \mathbb{N}_0$ by
\begin{align*}
k_0&=0,\\
k_m&=
\begin{cases}
\min\{k>k_{m-1}:\ell_k\neq \ell_{k_{m-1}}\},&\text{if the set is nonempty},\\
k_{m-1}+1,&\text{otherwise},
\end{cases}
\quad m\in \mathbb{N}.
\end{align*}
Then $k_0<k_1<\cdots$ and $(\ell_{k_m})_{m\in\mathbb{N}_0}$ is increasing.
Moreover, whenever $\ell_{k_{m-1}}<\ell_{k_m}$ and $\ell_{k_{m-1}}\le x<\ell_{k_m}$, we have $\bar k(x)=k_m-1$.
Hence
\begin{align}
\label{eq:distortion:inc:min}
    \varrho_h(C)
    &=\ell_{k_0}+\sum_{m\in\mathbb{N}}(\ell_{k_m}-\ell_{k_{m-1}})\,h\left(p_{k_m}(C)\right)\notag\\
    &=z_0+\sum_{m\in\mathbb{N}}z_m\,h(C(u_m,v_m)),
\end{align}
where $z_0=\ell_{k_0}=\ell_0\ge 0$, $z_m=\ell_{k_m}-\ell_{k_{m-1}}\ge 0$, and
$(u_m,v_m)=(\bar F(k_m),\bar G(k_m))\in [0,1]^2$
for $m\in\mathbb{N}$.
Since $(z_m)_{m\in\mathbb{N}_0}$ and $(u_m,v_m)$, $m\in\mathbb{N}$, do not depend on $C$ and $h$ is increasing, we conclude that $C\mapsto \varrho_h(C)$ is increasing w.r.t.~$\preceq$.

When $L=g(K_{\vee})$ for an increasing function $g$, we obtain
\begin{align*}
    \varrho_h(C)
    =\ell_0+\int_{[\ell_0,\ell_\infty)} h\left(q_{\bar k(x)+1}(C)\right)\,\rd x
\end{align*}
by analogous discussion.
Therefore, we have that
\begin{align}
\label{eq:distortion:inc:max}
    \varrho_h(C)
    &=\ell_{k_0}+\sum_{m\in\mathbb{N}}(\ell_{k_m}-\ell_{k_{m-1}})\,h\left(q_{k_m}(C)\right)\notag\\
    &=z_0+\sum_{m\in\mathbb{N}}z_m\,h(u_m+v_m-C(u_m,v_m)),
\end{align}
from which $C\mapsto \varrho_h(C)$ is decreasing.
Note that $(z_m)_{m\in\mathbb{N}_0}$ and $(u_m,v_m)$, $m\in\mathbb{N}$, in~\eqref{eq:distortion:inc:max} are the same as those used in~\eqref{eq:distortion:inc:min}.

Next, suppose that $L=g(K_{\wedge})$ for a decreasing function $g$.
Let
\begin{align*}
    \ell_k&=g(k),\quad k\in\mathbb{N}_0,\\
    \ell_{\infty}&=\lim_{k\to \infty}g(k)\in [0,\infty).
\end{align*}
If $x>\ell_0$, then $\mathbb{P}(L\le x)=1$.
For $\ell_{\infty}<x\le \ell_0$, let
\begin{align*}
    \underline k=\underline k(x)=\min\{k \in \mathbb{N}_0:\ell_k \le x\}.
\end{align*}
Then
\begin{align*}
    \mathbb{P}(L\le x)
    &=\sum_{k\in\mathbb{N}_0}\mathbb{P}(L\le x\mid K_{\wedge}=k)\mathbb{P}(K_{\wedge}=k)\\
    &=\sum_{k\in\mathbb{N}_0}\id_{\{\ell_k\le x\}}\{p_k(C)-p_{k+1}(C)\}\\
    &=\sum_{k=\underline k(x)}^{\infty}\{p_k(C)-p_{k+1}(C)\}\\
    &=p_{\underline k(x)}(C).
\end{align*}
For $0\le x<\ell_{\infty}$, we have $\mathbb{P}(L>x)=1$.
Therefore,
\begin{align*}
    \varrho_h(C)
    =\int_0^\infty h(\mathbb{P}(L>x))\,\rd x
=\ell_{\infty}+\int_{(\ell_\infty,\ell_0]} h\left(1-p_{\underline k(x)}(C)\right)\,\rd x.
\end{align*}

Recursively define the increasing sequence $k_0,k_1,\dots \in \mathbb{N}_0$ by
\begin{align*}
k_0&=0,\\
k_m&=
\begin{cases}
\min\{k>k_{m-1}:\ell_k\neq \ell_{k_{m-1}}\},&\text{if the set is nonempty},\\
k_{m-1}+1,&\text{otherwise},
\end{cases}
\quad m\in \mathbb{N}.
\end{align*}
Then $k_0<k_1<\cdots$ and $(\ell_{k_m})_{m\in\mathbb{N}_0}$ is decreasing.
Moreover, whenever $\ell_{k_m}<\ell_{k_{m-1}}$ and $\ell_{k_m}\le x<\ell_{k_{m-1}}$, we have $\underline k(x)=k_m$.
Therefore, by taking $z_0=\ell_{\infty}$, $z_m=\ell_{k_{m-1}}-\ell_{k_m}\ge 0$, and
$(u_m,v_m)=(\bar F(k_m),\bar G(k_m))\in [0,1]^2$
for $m\in\mathbb{N}$, we have that
\begin{align}
\label{eq:distortion:dec:min}
    \varrho_h(C)
    &=\ell_{\infty}+\sum_{m\in\mathbb{N}}(\ell_{k_{m-1}}-\ell_{k_m})\,h\left(1-p_{k_m}(C)\right)\notag\\
    &=z_0+\sum_{m\in\mathbb{N}}z_m\,h(1-C(u_m,v_m)),
\end{align}
which is decreasing.

With the same $(z_m)_{m\in\mathbb{N}_0}$ and $(u_m,v_m)$, $m\in\mathbb{N}$, used in~\eqref{eq:distortion:dec:min}, if $L=g(K_{\vee})$, then
\begin{align}
\label{eq:distortion:dec:max}
    \varrho_h(C)
    &=\ell_{\infty}+\sum_{m\in\mathbb{N}}(\ell_{k_{m-1}}-\ell_{k_m})\,h\left(1-q_{k_m}(C)\right)\notag\\
    &=z_0+\sum_{m\in\mathbb{N}}z_m\,h(1-u_m-v_m+C(u_m,v_m)).
\end{align}
Therefore, $C\mapsto \varrho_h(C)$ is increasing.

For all the four cases, we set $(u_m,v_m)=(\bar F(k_m),\bar G(k_m))$, $m\in\mathbb{N}$.
Since $(k_m)_{m\in\mathbb{N}_0}$ is an increasing sequence and $\bar F$ and $\bar G$ are decreasing functions, we have that
$\mathcal S=\{(u_m,v_m):m\in\mathbb{N}\}$ is an increasing set with
$u_1\ge u_2\ge \cdots$ and $v_1\ge v_2\ge \cdots$.

Next, suppose that human lifespan is limited.
Since $k_m\to\infty$ as $m\to\infty$ and $\bar F(t)=0$ and $\bar G(t)=0$ for all sufficiently large $t$, the following is well-defined:
\begin{align*}
    \bar m_\wedge = \min\{m\in\mathbb{N}:u_m=0\text{ or }v_m=0\}.
\end{align*}
For every $m \ge \bar m_\wedge$, we have that
\begin{align*}
    h(C(u_m,v_m))=0
    \quad\text{and}\quad
    h(1-C(u_m,v_m))=1.
\end{align*}
Therefore, the quantity in~\eqref{eq:distortion:inc:min} remains unchanged by setting $z_m=0$ for all $m\ge \bar m_\wedge$, and this is also the case in~\eqref{eq:distortion:dec:min} by setting $z_m=0$ for all $m\ge \bar m_\wedge$ and replacing $z_0=\ell_\infty$ with
\begin{align*}
    \tilde z_0=\ell_\infty+\sum_{m=\bar m_\wedge}^{\infty} z_m.
\end{align*}

Similarly, define
\begin{align*}
    \bar m_\vee = \min\{m\in\mathbb{N}:u_m=0\text{ and }v_m=0\}.
\end{align*}
For every $m \ge \bar m_\vee$, we have that
\begin{align*}
    h(u_m+v_m-C(u_m,v_m))=0
    \quad\text{and}\quad
    h(1-u_m-v_m+C(u_m,v_m))=1.
\end{align*}
Therefore, the quantity in~\eqref{eq:distortion:inc:max} remains unchanged by setting $z_m=0$ for all $m\ge \bar m_\vee$, and this is also the case in~\eqref{eq:distortion:dec:max} by setting $z_m=0$ for all $m\ge \bar m_\vee$ and replacing $z_0=\ell_\infty$ with
\begin{align*}
    \tilde z_0=\ell_\infty+\sum_{m=\bar m_\vee}^{\infty} z_m.
\end{align*}

Consequently, in all the cases, the formulas reduce to finite sums.
\end{proof}

\begin{proof}[Proof of Proposition~\ref{prop:bounds:rho:h}]
Let
\begin{align*}
 H(\mathbf{r})=z_0+\sum_{m=1}^{\bar m} z_m h(r_m),\quad \mathbf{r}\in\mathbb{R}^{\bar m},
\end{align*}
where we extend the domain of $h$ from $[0,1]$ to $\mathbb{R}$ by $h(\beta)=0$ for $\beta <0$,  and $h(\beta)=1$ for $\beta>1$.
Then 
$\overline\varrho_h(\mathcal C_{\mathcal S,\epsilon}(C^{\operatorname{ref}}))=\sup\{H(\mathbf{r}):\mathbf{r}\in \tilde{\mathcal R}(\epsilon)\}$,
where $\tilde{\mathcal R}(\epsilon)$ is the collection of $\mathbf{r}\in \mathbb{R}^{\bar m}$ such that
\begin{align}\label{eq:r:in:C}
r_m=A_m + B_m C(u_m,v_m),\quad m=1,\dots,\bar m, 
\end{align}
for some $C\in \mathcal C_{\mathcal S,\epsilon}(C^{\operatorname{ref}})$.
For a given $\mathbf{r}\in \mathbb{R}^{\bar m}$, let $\boldsymbol{\theta}=(\theta_1,\dots,\theta_{\bar m})^\top\in\mathbb{R}^{\bar m}$ with $\theta_m=(r_m-A_m)/B_m$.
Then there exists $C\in \mathcal C_{\mathcal S,\epsilon}(C^{\operatorname{ref}})$ such that
$\theta_m=C(u_m,v_m)$, $m=1,\dots,\bar m$, if and only if $\boldsymbol{\theta}$ satisfies the following conditions:
  \begin{align}
\label{eq:cond:monotone}    
&0\le \theta_{\bar m}\le \cdots \le \theta_{1}\le1,\\
\label{eq:cond:rev}   
& u_{\bar m}+v_{\bar m}-\theta_{\bar m}\le \cdots \le u_1+v_1-\theta_1,\\
\label{eq:cond:bounds}     
&W(u_m,v_m)\le \theta_m\le M(u_m,v_m),\quad m=1,\dots,\bar m,\\
\label{eq:cond:ball}
&\|\boldsymbol{\theta}-\boldsymbol{\theta}^{\operatorname{ref}}\|\le \epsilon,\quad \text{where $\theta_m^{\operatorname{ref}}=C^{\operatorname{ref}}(u_m,v_m)$, $m=1,\dots,\bar m$}.
\end{align}
Conditions~\eqref{eq:cond:monotone},~\eqref{eq:cond:rev} and~\eqref{eq:cond:bounds} follow from~\citet{genest1999characterization} 
on the existence of a copula $C\in \mathcal C$ with given values on $\mathcal S$, and Condition~\eqref{eq:cond:ball} is the restriction such that this copula is in the $\epsilon$-ball around $C^{\operatorname{ref}}$.
Note that $0\le \theta_{\bar m}$ and $\theta_1 \leq 1$ in~Condition~\eqref{eq:cond:monotone} are implied from Condition~\eqref{eq:cond:bounds}.

Let $\mathbf{e}_1,\dots,\mathbf{e}_{\bar m}$ be $\bar m$-dimensional standard unit vectors.
Then Condition~\eqref{eq:cond:monotone}, excluding $0\le \theta_{\bar m}$ and $\theta_1 \leq 1$, can be reformulated into 
\begin{align*}
(\mathbf{e}_2-\mathbf{e}_1,
\mathbf{e}_3-\mathbf{e}_2,\dots,
\mathbf{e}_{\bar m}-\mathbf{e}_{\bar m-1}
)^\top\boldsymbol{\theta}
\le \mathbf{0}_{\bar m - 1}.
\end{align*}
Next, Condition~\eqref{eq:cond:rev} is equivalent to
\begin{align*}
(
\mathbf{e}_1-\mathbf{e}_2,
\mathbf{e}_2-\mathbf{e}_3,\dots,
\mathbf{e}_{\bar m - 1}-\mathbf{e}_{\bar m})^\top\boldsymbol{\theta}+\begin{pmatrix}
    u_2+v_2 - (u_1+v_1)\\
    \vdots\\
    u_{\bar m }+v_{\bar m }-(u_{\bar m -1}+v_{\bar m-1})
\end{pmatrix}\le \mathbf{0}_{\bar m - 1}.
\end{align*}
Let $I_{\bar m}$ be the 
$\bar m$-dimensional unit diagonal matrix.
Then Condition~\eqref{eq:cond:bounds} is represented by:
\begin{align*}
    I_{\bar m} \boldsymbol{\theta}+\begin{pmatrix}
        -M(u_1,v_1)\\
        \vdots\\
        -M(u_{\bar m},v_{\bar m})\\
    \end{pmatrix}\le \mathbf{0}_{\bar m}\quad\text{and}\quad 
    - I_{\bar m} \boldsymbol{\theta}+\begin{pmatrix}
        W(u_1,v_1)\\
        \vdots\\
        W(u_{\bar m},v_{\bar m})\\
    \end{pmatrix}\le \mathbf{0}_{\bar m}.\\
\end{align*}
Since $\boldsymbol{\theta}$ is linear in $\mathbf{r}=(r_1,\dots,r_{\bar m})^\top$, the set of conditions~\eqref{eq:cond:monotone},~\eqref{eq:cond:rev} and~\eqref{eq:cond:bounds} can be written as 
$\bigcap_{j=1}^{4\bar m-2}\{\mathbf{r}\in \mathbb{R}^{\bar m}: K_j(\mathbf{r})\le 0\}$ for some linear functions $K_j:\mathbb{R}^{\bar m}\rightarrow \mathbb{R}$, $j=1,\dots,4\bar m-2$.
Finally, Condition~\eqref{eq:cond:ball} is equivalent to $K_{4\bar m - 1}(\mathbf{r})\le 0$ with $K_{4\bar m - 1}(\mathbf{r})=\|\mathbf{r}-\mathbf{r}^{\operatorname{ref}}\|-\epsilon$.

Consequently, we have that $\tilde{\mathcal R}(\epsilon)=\mathcal R(\epsilon)$ with  $\mathcal R(\epsilon)=\bigcap_{j=1}^{4\bar m - 1}\{\mathbf{r}\in \mathbb{R}^{\bar m}: K_j(\mathbf{r})\le 0\}$ for some convex functions $K_j:\mathbb{R}^{\bar m}\rightarrow \mathbb{R}$, $j=1,\dots,4\bar m - 1$.
Note that $\mathcal R(\epsilon)\subset[0,1]^{\bar m}$, which can be checked straightforwardly from the Fr{\'e}chet-Hoeffding inequality.
\end{proof}

\begin{proof}[Proof of Corollary~\ref{cor:bar:epsilon}]
    The representation~\eqref{eq:bar:epsilon:ref} follows by an analogous discussion with that in the proof of Proposition~\ref{prop:bounds:rho:h}.
    The set $\mathcal R$ is given by  $\mathcal R=\bigcap_{j=1}^{4\bar m-2}\{\mathbf{r}\in \mathbb{R}^{\bar m}: K_j(\mathbf{r})\le 0\}$ where $K_j$, $j=1,\dots,4\bar m-2$, are as defined in the proof of Proposition~\ref{prop:bounds:rho:h} with $(A_m,B_m)=(0,1)$, $m=1,\dots,\bar m$.
Together with the continuity of $\mathbf{r}\mapsto \|\mathbf{r}-\mathbf{r}^{\operatorname{ref}}\|$, 
    the supremum in~\eqref{eq:bar:epsilon} is attainable.
\end{proof}

\begin{proof}[Proof of Proposition~\ref{prop:var:bounds}]
For convenience, write
\begin{align*}
H(\mathbf{r})=z_0+\sum_{m=1}^{\bar m}z_m h_{\alpha}^{\operatorname{VaR}}(r_m),\quad \mathbf{r}=(r_1,\dots,r_{\bar m})^\top\in \mathbb{R}^{\bar m}.
\end{align*}
Here and below, a sum over an empty index set is understood to be $0$.

First, suppose that $g$ is increasing. In this case, 
$r_{\bar m}\le \cdots \le r_{1}$ holds for any $\mathbf{r}\in\mathcal R(\epsilon)$, and thus $\overline r_{\bar m}\le \cdots \le \overline r_{1}$ and $\underline r_{\bar m}\le \cdots \le \underline r_{1}$.

If $m_{\operatorname{L}}\le \bar m$, then there exists $\mathbf{r}^{\ast}\in \mathcal R(\epsilon)$ such that $r_{m_{\operatorname{L}}}^{\ast}\le 1-\alpha$.
Since $r_{\bar m}^{\ast}\le \cdots \le r_{m_{\operatorname{L}}}^{\ast}$, we have that
$r_{m_{\operatorname{L}}}^{\ast},\dots,r_{\bar m}^{\ast}\le 1-\alpha$.
Moreover, if $m_{\operatorname{L}}\ge 2$, then $\underline r_{m_{\operatorname{L}}-1}>1-\alpha$, and hence
$r_{m_{\operatorname{L}}-1}^{\ast},\dots,r_1^{\ast}>1-\alpha$.
Therefore, we have that
\begin{align*}
H(\mathbf{r}^{\ast})=\sum_{m=0}^{m_{\operatorname{L}}-1}z_m,
\end{align*}
where the case $m_{\operatorname{L}}=1$ is included.
Hence
\begin{align*}
\underline{\operatorname{VaR}}_\alpha(\mathcal C_{\mathcal S,\epsilon}(C^{\operatorname{ref}}))
=\inf\{H(\mathbf{r}):\mathbf{r}\in\mathcal R(\epsilon)\}
\le \sum_{m=0}^{m_{\operatorname{L}}-1}z_m.
\end{align*}
If $m_{\operatorname{L}}=\bar m+1$, then $r_1,\dots,r_{\bar m}>1-\alpha$ for every $\mathbf{r}\in\mathcal R(\epsilon)$.
Therefore, for every $\mathbf{r}\in\mathcal R(\epsilon)$,
\begin{align*}
H(\mathbf{r})=\sum_{m=0}^{\bar m}z_m=\sum_{m=0}^{m_{\operatorname{L}}-1}z_m,
\end{align*}
and thus
\begin{align*}
\underline{\operatorname{VaR}}_\alpha(\mathcal C_{\mathcal S,\epsilon}(C^{\operatorname{ref}}))
=\sum_{m=0}^{m_{\operatorname{L}}-1}z_m.
\end{align*}
If $m_{\operatorname{L}}=1$, then
$\underline{\operatorname{VaR}}_\alpha(\mathcal C_{\mathcal S,\epsilon}(C^{\operatorname{ref}}))=z_0$
since $H(\mathbf{r})\ge z_0$ for all $\mathbf{r}\in\mathcal R(\epsilon)$.
Suppose that $m_{\operatorname{L}}\ge 2$.
Since $\underline r_{m_{\operatorname{L}}-1}>1-\alpha$, we have that
$r_{m_{\operatorname{L}}-1}, \cdots, r_{1}>1-\alpha$ for every $\mathbf{r}\in\mathcal R(\epsilon)$.
Therefore, we have that
\begin{align*}
H(\mathbf{r})\ge \sum_{m=0}^{m_{\operatorname{L}}-1}z_m
\end{align*}
for every $\mathbf{r}\in\mathcal R(\epsilon)$.
Therefore, for $m_{\operatorname{L}}\le \bar m$ we have that
\begin{align*}
\underline{\operatorname{VaR}}_\alpha(\mathcal C_{\mathcal S,\epsilon}(C^{\operatorname{ref}}))
=\sum_{m=0}^{m_{\operatorname{L}}-1}z_m.
\end{align*}

If $m_{\operatorname{U}}\ge 1$, then there exists $\mathbf{r}^{\ast}\in \mathcal R(\epsilon)$ such that
$r_{m_{\operatorname{U}}}^{\ast}> 1-\alpha$.
Since $r_{m_{\operatorname{U}}}^{\ast}\le \cdots \le r_1^{\ast}$, we have that
$r_1^{\ast},\dots,r_{m_{\operatorname{U}}}^{\ast}>1-\alpha$.
Moreover, if $m_{\operatorname{U}}\le \bar m-1$, then $\overline r_{m_{\operatorname{U}}+1}\le 1-\alpha$, and hence
$r_{m_{\operatorname{U}}+1}^{\ast},\dots,r_{\bar m}^{\ast}\le 1-\alpha$.
Therefore, we have that
\begin{align*}
H(\mathbf{r}^{\ast})=\sum_{m=0}^{m_{\operatorname{U}}}z_m,
\end{align*}
where the case $m_{\operatorname{U}}=\bar m$ is included.
Hence
\begin{align*}
\overline{\operatorname{VaR}}_\alpha(\mathcal C_{\mathcal S,\epsilon}(C^{\operatorname{ref}}))
=\sup\{H(\mathbf{r}):\mathbf{r}\in\mathcal R(\epsilon)\}
\ge \sum_{m=0}^{m_{\operatorname{U}}}z_m.
\end{align*}
If $m_{\operatorname{U}}=0$, then $r_1,\dots,r_{\bar m}\le 1-\alpha$ for every $\mathbf{r}\in\mathcal R(\epsilon)$.
Therefore, for every $\mathbf{r}\in\mathcal R(\epsilon)$,
\begin{align*}
H(\mathbf{r})=z_0=\sum_{m=0}^{m_{\operatorname{U}}}z_m,
\end{align*}
and thus
\begin{align*}
\overline{\operatorname{VaR}}_\alpha(\mathcal C_{\mathcal S,\epsilon}(C^{\operatorname{ref}}))
=\sum_{m=0}^{m_{\operatorname{U}}}z_m.
\end{align*}
If $m_{\operatorname{U}}=\bar m$, then
$\overline{\operatorname{VaR}}_\alpha(\mathcal C_{\mathcal S,\epsilon}(C^{\operatorname{ref}}))
=\sum_{m=0}^{m_{\operatorname{U}}}z_m$
since $H(\mathbf{r})\le \sum_{m=0}^{m_{\operatorname{U}}}z_m$ for all $\mathbf{r}\in\mathcal R(\epsilon)$.
Suppose that $m_{\operatorname{U}}\le \bar m-1$.
Since $\overline r_{m_{\operatorname{U}}+1}\le 1-\alpha$, we have that
$r_{\bar m}, \cdots, r_{m_{\operatorname{U}}+1}\le 1-\alpha$ for every $\mathbf{r}\in\mathcal R(\epsilon)$.
Therefore, we have that
\begin{align*}
H(\mathbf{r})\le \sum_{m=0}^{m_{\operatorname{U}}}z_m
\end{align*}
for every $\mathbf{r}\in\mathcal R(\epsilon)$.
Consequently, for $m_{\operatorname{U}}\ge 1$ we have that
\begin{align*}
\overline{\operatorname{VaR}}_\alpha(\mathcal C_{\mathcal S,\epsilon}(C^{\operatorname{ref}}))
=\sum_{m=0}^{m_{\operatorname{U}}}z_m.
\end{align*}

Next, suppose that $g$ is decreasing.
The bounds can be derived essentially in the same way as above, but we provide details for completeness.
In this case, $r_1\le \cdots \le r_{\bar m}$ holds for any $\mathbf{r}\in\mathcal R(\epsilon)$ from~\eqref{eq:r:in:C} in the proof of Proposition~\ref{prop:bounds:rho:h}.
This implies that $\overline r_1\le \cdots \le \overline r_{\bar m}$ and $\underline r_1\le \cdots \le \underline r_{\bar m}$.

If $m_{\operatorname{L}}\ge 1$, then there exists $\mathbf{r}^{\ast}\in \mathcal R(\epsilon)$ such that $r_{m_{\operatorname{L}}}^{\ast}\le 1-\alpha$.
Since $r_1^{\ast}\le \cdots \le r_{m_{\operatorname{L}}}^{\ast}$, we have that
$r_1^{\ast},\dots,r_{m_{\operatorname{L}}}^{\ast}\le 1-\alpha$.
Moreover, if $m_{\operatorname{L}}\le \bar m-1$, then $\underline r_{m_{\operatorname{L}}+1}>1-\alpha$, and hence
$r_{m_{\operatorname{L}}+1}^{\ast},\dots,r_{\bar m}^{\ast}>1-\alpha$.
Therefore, we have that
\begin{align*}
H(\mathbf{r}^{\ast})=z_0+\sum_{m=m_{\operatorname{L}}+1}^{\bar m}z_m,
\end{align*}
where the case $m_{\operatorname{L}}=\bar m$ is included.
Hence
\begin{align*}
\underline{\operatorname{VaR}}_\alpha(\mathcal C_{\mathcal S,\epsilon}(C^{\operatorname{ref}}))
=\inf\{H(\mathbf{r}):\mathbf{r}\in\mathcal R(\epsilon)\}
\le z_0+\sum_{m=m_{\operatorname{L}}+1}^{\bar m}z_m.
\end{align*}
If $m_{\operatorname{L}}=0$, then $r_1,\dots,r_{\bar m}>1-\alpha$ for every $\mathbf{r}\in\mathcal R(\epsilon)$.
Therefore, for every $\mathbf{r}\in\mathcal R(\epsilon)$,
\begin{align*}
H(\mathbf{r})=z_0+\sum_{m=1}^{\bar m}z_m
=z_0+\sum_{m=m_{\operatorname{L}}+1}^{\bar m}z_m,
\end{align*}
and thus
\begin{align*}
\underline{\operatorname{VaR}}_\alpha(\mathcal C_{\mathcal S,\epsilon}(C^{\operatorname{ref}}))
=z_0+\sum_{m=m_{\operatorname{L}}+1}^{\bar m}z_m.
\end{align*}
If $m_{\operatorname{L}}=\bar m$, then
$\underline{\operatorname{VaR}}_\alpha(\mathcal C_{\mathcal S,\epsilon}(C^{\operatorname{ref}}))=z_0$
since $H(\mathbf{r})\ge z_0$ for all $\mathbf{r}\in\mathcal R(\epsilon)$.
Suppose that $m_{\operatorname{L}}\le \bar m-1$.
Since $\underline r_{m_{\operatorname{L}}+1}>1-\alpha$, we have that
$r_{m_{\operatorname{L}}+1},\dots,r_{\bar m}>1-\alpha$ for every $\mathbf{r}\in\mathcal R(\epsilon)$.
Therefore, we have that
\begin{align*}
H(\mathbf{r})\ge z_0+\sum_{m=m_{\operatorname{L}}+1}^{\bar m}z_m
\end{align*}
for every $\mathbf{r}\in\mathcal R(\epsilon)$.
For $m_{\operatorname{L}}\ge 1$, this leads to
\begin{align*}
\underline{\operatorname{VaR}}_\alpha(\mathcal C_{\mathcal S,\epsilon}(C^{\operatorname{ref}}))
=z_0+\sum_{m=m_{\operatorname{L}}+1}^{\bar m}z_m.
\end{align*}

If $m_{\operatorname{U}}\le \bar m$, then there exists $\mathbf{r}^{\ast}\in \mathcal R(\epsilon)$ such that
$r_{m_{\operatorname{U}}}^{\ast}> 1-\alpha$.
Since $r_{m_{\operatorname{U}}}^{\ast}\le \cdots \le r_{\bar m}^{\ast}$, we have that
$r_{m_{\operatorname{U}}}^{\ast},\dots,r_{\bar m}^{\ast}>1-\alpha$.
Moreover, if $m_{\operatorname{U}}\ge 2$, then $\overline r_{m_{\operatorname{U}}-1}\le 1-\alpha$, and hence
$r_1^{\ast},\dots,r_{m_{\operatorname{U}}-1}^{\ast}\le 1-\alpha$.
Therefore, we have that
\begin{align*}
H(\mathbf{r}^{\ast})=z_0+\sum_{m=m_{\operatorname{U}}}^{\bar m}z_m,
\end{align*}
where the case $m_{\operatorname{U}}=1$ is included.
Hence
\begin{align*}
\overline{\operatorname{VaR}}_\alpha(\mathcal C_{\mathcal S,\epsilon}(C^{\operatorname{ref}}))
=\sup\{H(\mathbf{r}):\mathbf{r}\in\mathcal R(\epsilon)\}
\ge z_0+\sum_{m=m_{\operatorname{U}}}^{\bar m}z_m.
\end{align*}
If $m_{\operatorname{U}}=\bar m+1$, then $r_1,\dots,r_{\bar m}\le 1-\alpha$ for every $\mathbf{r}\in\mathcal R(\epsilon)$.
Therefore, for every $\mathbf{r}\in\mathcal R(\epsilon)$,
\begin{align*}
H(\mathbf{r})=z_0=z_0+\sum_{m=m_{\operatorname{U}}}^{\bar m}z_m,
\end{align*}
and thus
\begin{align*}
\overline{\operatorname{VaR}}_\alpha(\mathcal C_{\mathcal S,\epsilon}(C^{\operatorname{ref}}))
=z_0+\sum_{m=m_{\operatorname{U}}}^{\bar m}z_m.
\end{align*}
If $m_{\operatorname{U}}=1$, then
$\overline{\operatorname{VaR}}_\alpha(\mathcal C_{\mathcal S,\epsilon}(C^{\operatorname{ref}}))
=\sum_{m=0}^{\bar m}z_m$
since $H(\mathbf{r})\le \sum_{m=0}^{\bar m}z_m$ for all $\mathbf{r}\in\mathcal R(\epsilon)$.
Suppose that $m_{\operatorname{U}}\ge 2$.
Since $\overline r_{m_{\operatorname{U}}-1}\le 1-\alpha$, we have that
$r_1,\dots,r_{m_{\operatorname{U}}-1}\le 1-\alpha$ for every $\mathbf{r}\in\mathcal R(\epsilon)$.
Therefore, we have that
\begin{align*}
H(\mathbf{r})\le z_0+\sum_{m=m_{\operatorname{U}}}^{\bar m}z_m
\end{align*}
for every $\mathbf{r}\in\mathcal R(\epsilon)$.
For $m_{\operatorname{U}}\le \bar m$, we conclude that
\begin{align*}
\overline{\operatorname{VaR}}_\alpha(\mathcal C_{\mathcal S,\epsilon}(C^{\operatorname{ref}}))
=z_0+\sum_{m=m_{\operatorname{U}}}^{\bar m}z_m.
\end{align*}
\end{proof}

\begin{proof}[Proof of Proposition~\ref{prop:es:bounds}]
For convenience, write
\begin{align*}
H(\mathbf{r})&=z_0+\sum_{m=1}^{\bar m}z_m h_{\alpha}^{\operatorname{ES}}(r_m),\quad \mathbf{r}=(r_1,\dots,r_{\bar m})^\top\in \mathbb{R}^{\bar m}.
\end{align*}
For each $m=0,\dots,\bar m$, the function $\mathbf{r}\mapsto H(\mathbf{r})$ on $\mathcal R_m(\epsilon)$ is of Form~\eqref{eq:H:ES:I} if $g$ is increasing and of Form~\eqref{eq:H:ES:II} if $g$ is decreasing.
Since $\mathcal R(\epsilon)=\cup_{m=0}^{\bar m} \mathcal R_m(\epsilon)$, we have that
\begin{align*}
\underline{\operatorname{ES}}_\alpha(\mathcal C_{\mathcal S,\epsilon}(C^{\operatorname{ref}}))&=\inf\{H(\mathbf{r}): \mathbf{r}\in\mathcal R(\epsilon)\}\\
&=\min\left[\inf\{H(\mathbf{r}): \mathbf{r}\in\mathcal R_m(\epsilon)\}: m\in\{0,\dots,\bar m\}\right]\\
&=\min\left\{\underline H_m: m\in\{0,\dots,\bar m\}\right\}.
\end{align*}
Similarly, we have that
\begin{align*}
\overline{\operatorname{ES}}_\alpha(\mathcal C_{\mathcal S,\epsilon}(C^{\operatorname{ref}}))&=\sup\{H(\mathbf{r}): \mathbf{r}\in\mathcal R(\epsilon)\}\\
&=\max\left[\sup\{H(\mathbf{r}): \mathbf{r}\in\mathcal R_m(\epsilon)\}: m\in\{0,\dots,\bar m\}\right]\\
&=\max\left\{\overline H_m: m\in\{0,\dots,\bar m\}\right\}.
\end{align*}

Since $\mathcal R(\epsilon)\neq \emptyset$, there exists $m\in\{0,\dots,\bar m\}$ such that $\mathcal R_m(\epsilon)\neq \emptyset$. 
For this $m$, the bounds $\underline H_m$ and $\overline H_m$ are attainable since 
$\mathcal R_m(\epsilon)$ is compact and convex, and $H$ is linear on $\mathcal R_m(\epsilon)$.
Consequently, the bounds $\underline{\operatorname{ES}}_\alpha(\mathcal C_{\mathcal S,\epsilon}(C^{\operatorname{ref}}))$ and $\overline{\operatorname{ES}}_\alpha(\mathcal C_{\mathcal S,\epsilon}(C^{\operatorname{ref}}))$ are attainable. 
\end{proof}

\begin{proof}[Proof of Lemma~\ref{lem:lp:norm}]
\begin{enumerate}
\item[\ref{lem:item:i}]  Take $\|\cdot\|=\|\cdot\|_{1}$. 
By introducing auxiliary variables $\mathbf{s}=(s_1,\dots,s_{\bar m})^\top\in \mathbb{R}^{\bar m}$, Condition~\eqref{eq:cond:ball} in the proof of Proposition~\ref{prop:bounds:rho:h} is equivalent to:
\begin{align*}
    \mathbf{1}_{\bar m}^\top \mathbf{s}-\epsilon \le 0,\quad -\mathbf{s}\le \mathbf{0}_{\bar m},\quad
   (I_{\bar m},-I_{\bar m})\begin{pmatrix}
    \boldsymbol{\theta}\\
    \mathbf{s}
\\
\end{pmatrix}
   -\boldsymbol{\theta}^{\operatorname{ref}}\le \mathbf{0}_{\bar m},\quad 
    - (I_{\bar m},I_{\bar m})\begin{pmatrix}
    \boldsymbol{\theta}\\
    \mathbf{s}
\\
\end{pmatrix}+\boldsymbol{\theta}^{\operatorname{ref}}\le \mathbf{0}_{\bar m}. 
\end{align*}
Since the constraint $-\mathbf{s}\le \mathbf{0}_{\bar m}$ is implied from the latter two ones, $2\bar m + 1$ linear constraints are additionally imposed  in this case.

\item[\ref{lem:item:ii}]
Take $\|\cdot\|=\|\cdot\|_{\infty}$. 
Then Condition~\eqref{eq:cond:ball} reduces to
\begin{align*}
    \boldsymbol{\theta}+ (-\boldsymbol{\theta}^{\operatorname{ref}}-\epsilon\mathbf{1}_{\bar m})\le \mathbf{0}_{\bar m}\quad \text{and}\quad 
    -\boldsymbol{\theta}+ (\boldsymbol{\theta}^{\operatorname{ref}}-\epsilon\mathbf{1}_{\bar m})\le \mathbf{0}_{\bar m}.
\end{align*}
Therefore, $\mathcal R(\epsilon)$ yields an intersection of $6\bar m-2$ linear constraints.
\end{enumerate}
\end{proof}

\section{Descriptions of linear programs}\label{sec:appendix:programs}

In Section~\ref{sec:uncertainty:level}, we showed that computing the bounds on $\E$, $\operatorname{VaR}$ and $\operatorname{ES}$ reduces to combinations of linear programs (LPs).
Since the programs are described at an abstract level, this section is devoted to describing these linear programs more concretely.

Under the assumption of limited human lifespan, the first task for a given monotone contract is to identify the representation~\eqref{eq:canonical:form:rho:h:finite} following the proof of Proposition~\ref{prop:monotonicity:distortion}.
We denote the decision variables by $\mathbf{r}=(r_1,\dots,r_{\bar m})^\top \in \mathbb{R}^{\bar m}$.

\subsection{Linear constraints}

We describe the linear constraints for the standard case where $(A_m,B_m)=(0,1)$ for all $m=1,\dots,\bar m$.
We first introduce the following linear constraints and bounds:
\begin{equation}\label{eq:cond:copula:compatibility}
\begin{aligned}
  & r_{m+1} - r_m \le 0, \quad m=1,\dots,\bar m - 1,  \\
    & r_m - r_{m+1} \le (u_m + v_m) - (u_{m+1} + v_{m+1}), \quad m=1,\dots,\bar m - 1, \\
    & r_m \ge \max\left(0, u_m + v_m - 1 \right), \quad m=1,\dots,\bar m,  \\
     & r_m \le \min\left(u_m, v_m, 1 \right), \quad m=1,\dots,\bar m.
\end{aligned}
\end{equation}
For the case of $\mathcal{L}^\infty$-norm, we introduce
\begin{equation}\label{eq:cond:Linf}
\begin{aligned}
 &   r_m \le r_m^{\text{ref}} + \epsilon,\\
 &    r_m^{\text{ref}} - \epsilon\le r_m.
\end{aligned}
\end{equation}
For the case of $\mathcal{L}^1$-norm, we introduce the auxiliary variables $\mathbf{s}=(s_1,\dots,s_{\bar m})^\top \in \mathbb{R}^{\bar m}$ and consider
\begin{equation}\label{eq:cond:L1}
    \begin{aligned}
 & s_1+\cdots + s_{\bar m}\le \epsilon,\\
 & 0 \le s_m,\quad m=1,\dots,\bar m,\\
 & r_m  \le s_m +  r^{\text{ref}}_m, \quad m=1,\dots,\bar m,\\
 & r^{\text{ref}}_m -s_m \le  r_m, \quad m=1,\dots,\bar m.
\end{aligned}
\end{equation}

For other cases of $(A_m,B_m)$, we replace 
$r_m$ with $\theta_m = (r_m-A_m)/B_m$ and $r_m^{\text{ref}}$ with $\theta_m^{\text{ref}} = (r_m^{\text{ref}}-A_m)/B_m$, $m=1,\dots,\bar m$, in~\eqref{eq:cond:copula:compatibility},~\eqref{eq:cond:Linf} and~\eqref{eq:cond:L1}, and rewrite them for the new decision variables $r_1,\dots,r_{\bar m}$.

\subsection{Linear programming}

In this section, we describe LPs to compute the bounds introduced in the paper. 
We focus on the case of $\mathcal L^\infty$-norm.
For the case of $\mathcal L^1$-norm, it suffices to replace~\eqref{eq:cond:Linf} with~\eqref{eq:cond:L1} in the descriptions.

First, to compute $\overline \epsilon(\mathcal C)$, it is required to solve
\begin{align*}
\text{maximize / minimize}\quad r_m\quad 
\text{subject to~}\eqref{eq:cond:copula:compatibility},
\end{align*}
to find $\overline r_m$ and $\underline r_m$, respectively, for $m=1,\dots,\bar m$.
Next, to compute the bounds on the expectation, we solve the following LPs:
\begin{align*}
\text{maximize / minimize}\quad  z_0 + \sum_{m=1}^{\bar m} z_m r_m\quad 
\text{subject to~}\eqref{eq:cond:copula:compatibility}\text{ and }\eqref{eq:cond:Linf}.
\end{align*}
For computing the VaR bounds, we solve the following LPs, for $m=1,\dots,\bar m$, to find $\overline{r}_m$ and $\underline{r}_m$ in Proposition~\ref{prop:var:bounds}:
\begin{align*}
\text{maximize / minimize}\quad  r_m\quad 
\text{subject to~}\eqref{eq:cond:copula:compatibility}\text{ and }\eqref{eq:cond:Linf};
\end{align*}
see the proposition for the rest of calculation steps.
Finally, we describe LPs to compute the ES bounds. 
We first assume that $g$ is increasing.
Then, for $m=0,\dots,\bar m$, solve
\begin{align*}
&\text{maximize / minimize}\quad  \sum_{m'=0}^{m} z_{m'}  +  \sum_{m'=m+1}^{\bar m} \frac{z_{m'}}{1-\alpha}r_{m'}\\
&\qquad\text{subject to~}\eqref{eq:cond:copula:compatibility},~\eqref{eq:cond:Linf}\text{ and }
\begin{cases}
    r_1 \le 1-\alpha,& \text{ if $m=0$},\\
    r_{m+1} \le 1-\alpha\text{ and }1-\alpha \le r_m,& \text{ if $m=1,\dots,\bar m -1$},\\
    1-\alpha \le r_{\bar m},& \text{ if $m=\bar m$}.\\
\end{cases}
\end{align*}
If $g$ is decreasing, then solve, for $m=0,\dots,\bar m$,
\begin{align*}
&\text{maximize / minimize}\quad  z_0+\sum_{m'=m+1}^{\bar m} z_{m'}  +  \sum_{m'=1}^{m} \frac{z_{m'}}{1-\alpha}r_{m'}\\
&\qquad\text{subject to~}\eqref{eq:cond:copula:compatibility},~\eqref{eq:cond:Linf}\text{ and }
\begin{cases}
    1-\alpha \le r_1,& \text{ if $m=0$},\\
    r_{m} \le 1-\alpha\text{ and }1-\alpha \le r_{m+1},& \text{ if $m=1,\dots,\bar m -1$},\\
    r_{\bar m} \le 1-\alpha ,& \text{ if $m=\bar m$}.\\
\end{cases}
\end{align*}
The remaining calculation steps are provided in Proposition~\ref{prop:es:bounds}.




\end{appendices}



\begin{thebibliography}{}
\renewcommand{\doi}[1]{\url{https://doi.org/#1}}
\bibcommenthead

\bibitem [\protect \citeauthoryear {%
Alsina%
, Nelsen%
\BCBL {}\ \BBA {} Schweizer%
}{%
Alsina%
\ \protect \BOthers {.}}{%
{\protect \APACyear {1993}}%
}]{%
alsina1993characterization}
\APACinsertmetastar {%
alsina1993characterization}%
\begin{APACrefauthors}%
Alsina, C.%
, Nelsen, R.B.%
\BCBL {} Schweizer, B.%
\end{APACrefauthors}%
\unskip\
\newblock
\APACrefYearMonthDay{1993}{}{}.
\newblock
{\BBOQ}\APACrefatitle {On the characterization of a class of binary operations on distribution functions} {On the characterization of a class of binary operations on distribution functions}.{\BBCQ}
\newblock
\APACjournalVolNumPages{Statistics \& Probability Letters}{17}{2}{85--89,}
\newblock

\newblock

\PrintBackRefs{\CurrentBib}

\bibitem [\protect \citeauthoryear {%
Barrieu%
\ \BBA {} Scandolo%
}{%
Barrieu%
\ \BBA {} Scandolo%
}{%
{\protect \APACyear {2015}}%
}]{%
barrieu2015assessing}
\APACinsertmetastar {%
barrieu2015assessing}%
\begin{APACrefauthors}%
Barrieu, P.%
\BCBT {}\ \BBA {} Scandolo, G.%
\end{APACrefauthors}%
\unskip\
\newblock
\APACrefYearMonthDay{2015}{}{}.
\newblock
{\BBOQ}\APACrefatitle {Assessing financial model risk} {Assessing financial model risk}.{\BBCQ}
\newblock
\APACjournalVolNumPages{European Journal of Operational Research}{242}{2}{546--556,}
\newblock

\newblock

\PrintBackRefs{\CurrentBib}

\bibitem [\protect \citeauthoryear {%
Berge%
}{%
Berge%
}{%
{\protect \APACyear {1963}}%
}]{%
Berge1963}
\APACinsertmetastar {%
Berge1963}%
\begin{APACrefauthors}%
Berge, C.%
\end{APACrefauthors}%
\unskip\
\newblock
\APACrefYear{1963}.
\newblock
\APACrefbtitle {Topological Spaces} {Topological spaces}.
\newblock
\APACaddressPublisher{New York}{Macmillan}.
\PrintBackRefs{\CurrentBib}

\bibitem [\protect \citeauthoryear {%
Bernard%
, De~Vecchi%
\BCBL {}\ \BBA {} Vanduffel%
}{%
Bernard%
\ \protect \BOthers {.}}{%
{\protect \APACyear {2024}}%
}]{%
bernard2024robust}
\APACinsertmetastar {%
bernard2024robust}%
\begin{APACrefauthors}%
Bernard, C.%
, De~Vecchi, C.%
\BCBL {} Vanduffel, S.%
\end{APACrefauthors}%
\unskip\
\newblock
\APACrefYearMonthDay{2024}{}{}.
\newblock
{\BBOQ}\APACrefatitle {Robust assessment of life insurance products} {Robust assessment of life insurance products}.{\BBCQ}
\newblock
\APACjournalVolNumPages{Annals of Operations Research}{}{}{1--27,}
\newblock

\newblock

\PrintBackRefs{\CurrentBib}

\bibitem [\protect \citeauthoryear {%
Bernard%
, R{\"u}schendorf%
, Vanduffel%
\BCBL {}\ \BBA {} Wang%
}{%
Bernard%
\ \protect \BOthers {.}}{%
{\protect \APACyear {2017}}%
}]{%
bernard2017risk}
\APACinsertmetastar {%
bernard2017risk}%
\begin{APACrefauthors}%
Bernard, C.%
, R{\"u}schendorf, L.%
, Vanduffel, S.%
\BCBL {} Wang, R.%
\end{APACrefauthors}%
\unskip\
\newblock
\APACrefYearMonthDay{2017}{}{}.
\newblock
{\BBOQ}\APACrefatitle {Risk bounds for factor models} {Risk bounds for factor models}.{\BBCQ}
\newblock
\APACjournalVolNumPages{Finance and Stochastics}{21}{}{631--659,}
\newblock

\newblock

\PrintBackRefs{\CurrentBib}

\bibitem [\protect \citeauthoryear {%
Bignozzi%
, Puccetti%
\BCBL {}\ \BBA {} R{\"u}schendorf%
}{%
Bignozzi%
\ \protect \BOthers {.}}{%
{\protect \APACyear {2015}}%
}]{%
bignozzi2015reducing}
\APACinsertmetastar {%
bignozzi2015reducing}%
\begin{APACrefauthors}%
Bignozzi, V.%
, Puccetti, G.%
\BCBL {} R{\"u}schendorf, L.%
\end{APACrefauthors}%
\unskip\
\newblock
\APACrefYearMonthDay{2015}{}{}.
\newblock
{\BBOQ}\APACrefatitle {Reducing model risk via positive and negative dependence assumptions} {Reducing model risk via positive and negative dependence assumptions}.{\BBCQ}
\newblock
\APACjournalVolNumPages{Insurance: Mathematics and Economics}{61}{}{17--26,}
\newblock

\newblock

\PrintBackRefs{\CurrentBib}

\bibitem [\protect \citeauthoryear {%
Boyd%
\ \BBA {} Vandenberghe%
}{%
Boyd%
\ \BBA {} Vandenberghe%
}{%
{\protect \APACyear {2004}}%
}]{%
boyd2004convex}
\APACinsertmetastar {%
boyd2004convex}%
\begin{APACrefauthors}%
Boyd, S.%
\BCBT {}\ \BBA {} Vandenberghe, L.%
\end{APACrefauthors}%
\unskip\
\newblock
\APACrefYear{2004}.
\newblock
\APACrefbtitle {Convex Optimization} {Convex optimization}.
\newblock
\APACaddressPublisher{}{Cambridge University Press}.
\PrintBackRefs{\CurrentBib}

\bibitem [\protect \citeauthoryear {%
Carriere%
}{%
Carriere%
}{%
{\protect \APACyear {2000}}%
}]{%
carriere2000bivariate}
\APACinsertmetastar {%
carriere2000bivariate}%
\begin{APACrefauthors}%
Carriere, J.F.%
\end{APACrefauthors}%
\unskip\
\newblock
\APACrefYearMonthDay{2000}{}{}.
\newblock
{\BBOQ}\APACrefatitle {Bivariate survival models for coupled lives} {Bivariate survival models for coupled lives}.{\BBCQ}
\newblock
\APACjournalVolNumPages{Scandinavian Actuarial Journal}{2000}{1}{17--32,}
\newblock

\newblock

\PrintBackRefs{\CurrentBib}

\bibitem [\protect \citeauthoryear {%
Denuit%
\ \BBA {} Cornett%
}{%
Denuit%
\ \BBA {} Cornett%
}{%
{\protect \APACyear {1999}}%
}]{%
denuit1999multilife}
\APACinsertmetastar {%
denuit1999multilife}%
\begin{APACrefauthors}%
Denuit, M.%
\BCBT {}\ \BBA {} Cornett, A.%
\end{APACrefauthors}%
\unskip\
\newblock
\APACrefYearMonthDay{1999}{}{}.
\newblock
{\BBOQ}\APACrefatitle {Multilife Premium Calculation with Dependent Future Lifetimes} {Multilife premium calculation with dependent future lifetimes}.{\BBCQ}
\newblock
\APACjournalVolNumPages{Journal of Actuarial Practice, Volume 7, Nos. 1 and 2, 1999}{7}{}{147,}
\newblock

\newblock

\PrintBackRefs{\CurrentBib}

\bibitem [\protect \citeauthoryear {%
Denuit%
, Dhaene%
, Le~Bailly~de Tilleghem%
\BCBL {}\ \BBA {} Teghem%
}{%
Denuit%
\ \protect \BOthers {.}}{%
{\protect \APACyear {2001}}%
}]{%
denuit2001measuring}
\APACinsertmetastar {%
denuit2001measuring}%
\begin{APACrefauthors}%
Denuit, M.%
, Dhaene, J.%
, Le~Bailly~de Tilleghem, C.%
\BCBL {} Teghem, S.%
\end{APACrefauthors}%
\unskip\
\newblock
\APACrefYearMonthDay{2001}{}{}.
\newblock
{\BBOQ}\APACrefatitle {Measuring the impact of dependence among insured lifelengths} {Measuring the impact of dependence among insured lifelengths}.{\BBCQ}
\newblock
\APACjournalVolNumPages{Belgian Actuarial Bulletin}{1}{1}{18--39,}
\newblock

\newblock

\PrintBackRefs{\CurrentBib}

\bibitem [\protect \citeauthoryear {%
Dickson%
, Hardy%
\BCBL {}\ \BBA {} Waters%
}{%
Dickson%
\ \protect \BOthers {.}}{%
{\protect \APACyear {2019}}%
}]{%
dickson2019actuarial}
\APACinsertmetastar {%
dickson2019actuarial}%
\begin{APACrefauthors}%
Dickson, D.C.%
, Hardy, M.R.%
\BCBL {} Waters, H.R.%
\end{APACrefauthors}%
\unskip\
\newblock
\APACrefYear{2019}.
\newblock
\APACrefbtitle {Actuarial mathematics for life contingent risks} {Actuarial mathematics for life contingent risks}.
\newblock
\APACaddressPublisher{}{Cambridge University Press}.
\PrintBackRefs{\CurrentBib}

\bibitem [\protect \citeauthoryear {%
Dufresne%
, Hashorva%
, Ratovomirija%
\BCBL {}\ \BBA {} Toukourou%
}{%
Dufresne%
\ \protect \BOthers {.}}{%
{\protect \APACyear {2018}}%
}]{%
dufresne2018age}
\APACinsertmetastar {%
dufresne2018age}%
\begin{APACrefauthors}%
Dufresne, F.%
, Hashorva, E.%
, Ratovomirija, G.%
\BCBL {} Toukourou, Y.%
\end{APACrefauthors}%
\unskip\
\newblock
\APACrefYearMonthDay{2018}{}{}.
\newblock
{\BBOQ}\APACrefatitle {On age difference in joint lifetime modelling with life insurance annuity applications} {On age difference in joint lifetime modelling with life insurance annuity applications}.{\BBCQ}
\newblock
\APACjournalVolNumPages{Annals of Actuarial Science}{12}{2}{350--371,}
\newblock

\newblock

\PrintBackRefs{\CurrentBib}

\bibitem [\protect \citeauthoryear {%
Einmahl%
, Einmahl%
\BCBL {}\ \BBA {} de Haan%
}{%
Einmahl%
\ \protect \BOthers {.}}{%
{\protect \APACyear {2019}}%
}]{%
einmahl2019limits}
\APACinsertmetastar {%
einmahl2019limits}%
\begin{APACrefauthors}%
Einmahl, J.J.%
, Einmahl, J.H.%
\BCBL {} de Haan, L.%
\end{APACrefauthors}%
\unskip\
\newblock
\APACrefYearMonthDay{2019}{}{}.
\newblock
{\BBOQ}\APACrefatitle {Limits to human life span through extreme value theory} {Limits to human life span through extreme value theory}.{\BBCQ}
\newblock
\APACjournalVolNumPages{Journal of the American Statistical Association}{114}{527}{1075--1080,}
\newblock

\newblock

\PrintBackRefs{\CurrentBib}

\bibitem [\protect \citeauthoryear {%
Embrechts%
, Puccetti%
\BCBL {}\ \BBA {} R{\"u}schendorf%
}{%
Embrechts%
\ \protect \BOthers {.}}{%
{\protect \APACyear {2013}}%
}]{%
embrechts2013model}
\APACinsertmetastar {%
embrechts2013model}%
\begin{APACrefauthors}%
Embrechts, P.%
, Puccetti, G.%
\BCBL {} R{\"u}schendorf, L.%
\end{APACrefauthors}%
\unskip\
\newblock
\APACrefYearMonthDay{2013}{}{}.
\newblock
{\BBOQ}\APACrefatitle {Model uncertainty and VaR aggregation} {Model uncertainty and var aggregation}.{\BBCQ}
\newblock
\APACjournalVolNumPages{Journal of Banking \& Finance}{37}{8}{2750--2764,}
\newblock

\newblock

\PrintBackRefs{\CurrentBib}

\bibitem [\protect \citeauthoryear {%
Frees%
, Carriere%
\BCBL {}\ \BBA {} Valdez%
}{%
Frees%
\ \protect \BOthers {.}}{%
{\protect \APACyear {1996}}%
}]{%
frees1996annuity}
\APACinsertmetastar {%
frees1996annuity}%
\begin{APACrefauthors}%
Frees, E.W.%
, Carriere, J.%
\BCBL {} Valdez, E.%
\end{APACrefauthors}%
\unskip\
\newblock
\APACrefYearMonthDay{1996}{}{}.
\newblock
{\BBOQ}\APACrefatitle {Annuity valuation with dependent mortality.} {Annuity valuation with dependent mortality.}{\BBCQ}
\newblock
\APACjournalVolNumPages{Journal of risk and insurance}{63}{2}{229--261,}
\newblock

\newblock

\PrintBackRefs{\CurrentBib}

\bibitem [\protect \citeauthoryear {%
Genest%
, Molina%
, Lallena%
\BCBL {}\ \BBA {} Sempi%
}{%
Genest%
\ \protect \BOthers {.}}{%
{\protect \APACyear {1999}}%
}]{%
genest1999characterization}
\APACinsertmetastar {%
genest1999characterization}%
\begin{APACrefauthors}%
Genest, C.%
, Molina, J.Q.%
, Lallena, J.R.%
\BCBL {} Sempi, C.%
\end{APACrefauthors}%
\unskip\
\newblock
\APACrefYearMonthDay{1999}{}{}.
\newblock
{\BBOQ}\APACrefatitle {A characterization of quasi-copulas} {A characterization of quasi-copulas}.{\BBCQ}
\newblock
\APACjournalVolNumPages{Journal of Multivariate Analysis}{69}{2}{193--205,}
\newblock

\newblock

\PrintBackRefs{\CurrentBib}

\bibitem [\protect \citeauthoryear {%
Gobbi%
, Kolev%
\BCBL {}\ \BBA {} Mulinacci%
}{%
Gobbi%
\ \protect \BOthers {.}}{%
{\protect \APACyear {2019}}%
}]{%
gobbi2019joint}
\APACinsertmetastar {%
gobbi2019joint}%
\begin{APACrefauthors}%
Gobbi, F.%
, Kolev, N.%
\BCBL {} Mulinacci, S.%
\end{APACrefauthors}%
\unskip\
\newblock
\APACrefYearMonthDay{2019}{}{}.
\newblock
{\BBOQ}\APACrefatitle {Joint life insurance pricing using extended Marshall--Olkin models} {Joint life insurance pricing using extended marshall--olkin models}.{\BBCQ}
\newblock
\APACjournalVolNumPages{ASTIN Bulletin: The Journal of the IAA}{49}{2}{409--432,}
\newblock

\newblock

\PrintBackRefs{\CurrentBib}

\bibitem [\protect \citeauthoryear {%
Gotoh%
\ \BBA {} Uryasev%
}{%
Gotoh%
\ \BBA {} Uryasev%
}{%
{\protect \APACyear {2016}}%
}]{%
gotoh2016two}
\APACinsertmetastar {%
gotoh2016two}%
\begin{APACrefauthors}%
Gotoh, J.%
\BCBT {}\ \BBA {} Uryasev, S.%
\end{APACrefauthors}%
\unskip\
\newblock
\APACrefYearMonthDay{2016}{}{}.
\newblock
{\BBOQ}\APACrefatitle {Two pairs of families of polyhedral norms versus $\ell_p$-norms: proximity and applications in optimization} {Two pairs of families of polyhedral norms versus $\ell_p$-norms: proximity and applications in optimization}.{\BBCQ}
\newblock
\APACjournalVolNumPages{Mathematical Programming}{156}{1}{391--431,}
\newblock

\newblock

\PrintBackRefs{\CurrentBib}

\bibitem [\protect \citeauthoryear {%
Hougaard%
}{%
Hougaard%
}{%
{\protect \APACyear {2000}}%
}]{%
hougaard2000analysis}
\APACinsertmetastar {%
hougaard2000analysis}%
\begin{APACrefauthors}%
Hougaard, P.%
\end{APACrefauthors}%
\unskip\
\newblock
\APACrefYear{2000}.
\newblock
\APACrefbtitle {Analysis of Multivariate Survival Data} {Analysis of multivariate survival data}.
\newblock
\APACaddressPublisher{New York, NY}{Springer}.
\PrintBackRefs{\CurrentBib}

\bibitem [\protect \citeauthoryear {%
Hua%
\ \BBA {} Joe%
}{%
Hua%
\ \BBA {} Joe%
}{%
{\protect \APACyear {2011}}%
}]{%
hua2011tail}
\APACinsertmetastar {%
hua2011tail}%
\begin{APACrefauthors}%
Hua, L.%
\BCBT {}\ \BBA {} Joe, H.%
\end{APACrefauthors}%
\unskip\
\newblock
\APACrefYearMonthDay{2011}{}{}.
\newblock
{\BBOQ}\APACrefatitle {Tail order and intermediate tail dependence of multivariate copulas} {Tail order and intermediate tail dependence of multivariate copulas}.{\BBCQ}
\newblock
\APACjournalVolNumPages{Journal of Multivariate Analysis}{102}{10}{1454--1471,}
\newblock

\newblock

\PrintBackRefs{\CurrentBib}

\bibitem [\protect \citeauthoryear {%
Joe%
}{%
Joe%
}{%
{\protect \APACyear {2014}}%
}]{%
joe2014dependence}
\APACinsertmetastar {%
joe2014dependence}%
\begin{APACrefauthors}%
Joe, H.%
\end{APACrefauthors}%
\unskip\
\newblock
\APACrefYear{2014}.
\newblock
\APACrefbtitle {Dependence modeling with copulas} {Dependence modeling with copulas}.
\newblock
\APACaddressPublisher{Florida}{CRC Press}.
\PrintBackRefs{\CurrentBib}

\bibitem [\protect \citeauthoryear {%
Kaas%
, Laeven%
\BCBL {}\ \BBA {} Nelsen%
}{%
Kaas%
\ \protect \BOthers {.}}{%
{\protect \APACyear {2009}}%
}]{%
kaas2009worst}
\APACinsertmetastar {%
kaas2009worst}%
\begin{APACrefauthors}%
Kaas, R.%
, Laeven, R.J.%
\BCBL {} Nelsen, R.B.%
\end{APACrefauthors}%
\unskip\
\newblock
\APACrefYearMonthDay{2009}{}{}.
\newblock
{\BBOQ}\APACrefatitle {Worst VaR scenarios with given marginals and measures of association} {Worst var scenarios with given marginals and measures of association}.{\BBCQ}
\newblock
\APACjournalVolNumPages{Insurance: Mathematics and Economics}{44}{2}{146--158,}
\newblock

\newblock

\PrintBackRefs{\CurrentBib}

\bibitem [\protect \citeauthoryear {%
Koike%
, Kato%
\BCBL {}\ \BBA {} Hofert%
}{%
Koike%
\ \protect \BOthers {.}}{%
{\protect \APACyear {2023}}%
}]{%
koike2023measuring}
\APACinsertmetastar {%
koike2023measuring}%
\begin{APACrefauthors}%
Koike, T.%
, Kato, S.%
\BCBL {} Hofert, M.%
\end{APACrefauthors}%
\unskip\
\newblock
\APACrefYearMonthDay{2023}{}{}.
\newblock
{\BBOQ}\APACrefatitle {Measuring non-exchangeable tail dependence using tail copulas} {Measuring non-exchangeable tail dependence using tail copulas}.{\BBCQ}
\newblock
\APACjournalVolNumPages{ASTIN Bulletin: The Journal of the IAA}{53}{2}{466--487,}
\newblock

\newblock

\PrintBackRefs{\CurrentBib}

\bibitem [\protect \citeauthoryear {%
Lai%
\ \BBA {} Xie%
}{%
Lai%
\ \BBA {} Xie%
}{%
{\protect \APACyear {2006}}%
}]{%
lai2006stochastic}
\APACinsertmetastar {%
lai2006stochastic}%
\begin{APACrefauthors}%
Lai, C.D.%
\BCBT {}\ \BBA {} Xie, M.%
\end{APACrefauthors}%
\unskip\
\newblock
\APACrefYear{2006}.
\newblock
\APACrefbtitle {Stochastic ageing and dependence for reliability} {Stochastic ageing and dependence for reliability}.
\newblock
\APACaddressPublisher{}{Springer Science \& Business Media}.
\PrintBackRefs{\CurrentBib}

\bibitem [\protect \citeauthoryear {%
Liu%
\ \BBA {} Wang%
}{%
Liu%
\ \BBA {} Wang%
}{%
{\protect \APACyear {2017}}%
}]{%
liu2017collective}
\APACinsertmetastar {%
liu2017collective}%
\begin{APACrefauthors}%
Liu, H.%
\BCBT {}\ \BBA {} Wang, R.%
\end{APACrefauthors}%
\unskip\
\newblock
\APACrefYearMonthDay{2017}{}{}.
\newblock
{\BBOQ}\APACrefatitle {Collective risk models with dependence uncertainty} {Collective risk models with dependence uncertainty}.{\BBCQ}
\newblock
\APACjournalVolNumPages{ASTIN Bulletin: The Journal of the IAA}{47}{2}{361--389,}
\newblock

\newblock

\PrintBackRefs{\CurrentBib}

\bibitem [\protect \citeauthoryear {%
Lu%
}{%
Lu%
}{%
{\protect \APACyear {2017}}%
}]{%
lu2017broken}
\APACinsertmetastar {%
lu2017broken}%
\begin{APACrefauthors}%
Lu, Y.%
\end{APACrefauthors}%
\unskip\
\newblock
\APACrefYearMonthDay{2017}{}{}.
\newblock
{\BBOQ}\APACrefatitle {Broken-heart, common life, heterogeneity: Analyzing the spousal mortality dependence} {Broken-heart, common life, heterogeneity: Analyzing the spousal mortality dependence}.{\BBCQ}
\newblock
\APACjournalVolNumPages{ASTIN Bulletin: the Journal of the IAA}{47}{3}{837--874,}
\newblock

\newblock

\PrintBackRefs{\CurrentBib}

\bibitem [\protect \citeauthoryear {%
Mardani-Fard%
, Sadooghi-Alvandi%
\BCBL {}\ \BBA {} Shishebor%
}{%
Mardani-Fard%
\ \protect \BOthers {.}}{%
{\protect \APACyear {2010}}%
}]{%
mardani2010bounds}
\APACinsertmetastar {%
mardani2010bounds}%
\begin{APACrefauthors}%
Mardani-Fard, H.%
, Sadooghi-Alvandi, S.%
\BCBL {} Shishebor, Z.%
\end{APACrefauthors}%
\unskip\
\newblock
\APACrefYearMonthDay{2010}{}{}.
\newblock
{\BBOQ}\APACrefatitle {Bounds on bivariate distribution functions with given margins and known values at several points} {Bounds on bivariate distribution functions with given margins and known values at several points}.{\BBCQ}
\newblock
\APACjournalVolNumPages{Communications in Statistics—Theory and Methods}{39}{20}{3596--3621,}
\newblock

\newblock

\PrintBackRefs{\CurrentBib}

\bibitem [\protect \citeauthoryear {%
Nelsen%
}{%
Nelsen%
}{%
{\protect \APACyear {2006}}%
}]{%
nelsen2006introduction}
\APACinsertmetastar {%
nelsen2006introduction}%
\begin{APACrefauthors}%
Nelsen, R.B.%
\end{APACrefauthors}%
\unskip\
\newblock
\APACrefYear{2006}.
\newblock
\APACrefbtitle {An Introduction to Copulas} {An introduction to copulas}.
\newblock
\APACaddressPublisher{New York}{Springer}.
\PrintBackRefs{\CurrentBib}

\bibitem [\protect \citeauthoryear {%
Nelsen%
, Quesada-Molina%
, Rodri{\'\i}guez-Lallena%
\BCBL {}\ \BBA {} {\'U}beda-Flores%
}{%
Nelsen%
\ \protect \BOthers {.}}{%
{\protect \APACyear {2001}}%
}]{%
nelsen2001bounds}
\APACinsertmetastar {%
nelsen2001bounds}%
\begin{APACrefauthors}%
Nelsen, R.B.%
, Quesada-Molina, J.J.%
, Rodri{\'\i}guez-Lallena, J.A.%
\BCBL {} {\'U}beda-Flores, M.%
\end{APACrefauthors}%
\unskip\
\newblock
\APACrefYearMonthDay{2001}{}{}.
\newblock
{\BBOQ}\APACrefatitle {Bounds on bivariate distribution functions with given margins and measures of association} {Bounds on bivariate distribution functions with given margins and measures of association}.{\BBCQ}
\newblock
\APACjournalVolNumPages{Communications in Statistics-Theory and Methods}{30}{6}{1055--1062,}
\newblock

\newblock

\PrintBackRefs{\CurrentBib}

\bibitem [\protect \citeauthoryear {%
Parkes%
, Benjamin%
\BCBL {}\ \BBA {} Fitzgerald%
}{%
Parkes%
\ \protect \BOthers {.}}{%
{\protect \APACyear {1969}}%
}]{%
parkes1969broken}
\APACinsertmetastar {%
parkes1969broken}%
\begin{APACrefauthors}%
Parkes, C.M.%
, Benjamin, B.%
\BCBL {} Fitzgerald, R.G.%
\end{APACrefauthors}%
\unskip\
\newblock
\APACrefYearMonthDay{1969}{}{}.
\newblock
{\BBOQ}\APACrefatitle {Broken heart: a statistical study of increased mortality among widowers} {Broken heart: a statistical study of increased mortality among widowers}.{\BBCQ}
\newblock
\APACjournalVolNumPages{British Medical Journal}{1}{5646}{740--743,}
\newblock

\newblock

\PrintBackRefs{\CurrentBib}

\bibitem [\protect \citeauthoryear {%
Pichler%
}{%
Pichler%
}{%
{\protect \APACyear {2014}}%
}]{%
pichler2014insurance}
\APACinsertmetastar {%
pichler2014insurance}%
\begin{APACrefauthors}%
Pichler, A.%
\end{APACrefauthors}%
\unskip\
\newblock
\APACrefYearMonthDay{2014}{}{}.
\newblock
{\BBOQ}\APACrefatitle {Insurance pricing under ambiguity} {Insurance pricing under ambiguity}.{\BBCQ}
\newblock
\APACjournalVolNumPages{European Actuarial Journal}{4}{}{335--364,}
\newblock

\newblock

\PrintBackRefs{\CurrentBib}

\bibitem [\protect \citeauthoryear {%
Rudin%
}{%
Rudin%
}{%
{\protect \APACyear {1976}}%
}]{%
Rudin1976}
\APACinsertmetastar {%
Rudin1976}%
\begin{APACrefauthors}%
Rudin, W.%
\end{APACrefauthors}%
\unskip\
\newblock
\APACrefYear{1976}.
\newblock
\APACrefbtitle {Principles of Mathematical Analysis} {Principles of mathematical analysis}\ (\PrintOrdinal{3rd}\ \BEd).
\newblock
\APACaddressPublisher{New York}{McGraw-Hill}.
\PrintBackRefs{\CurrentBib}

\bibitem [\protect \citeauthoryear {%
R{\"u}schendorf%
, Vanduffel%
\BCBL {}\ \BBA {} Bernard%
}{%
R{\"u}schendorf%
\ \protect \BOthers {.}}{%
{\protect \APACyear {2024}}%
}]{%
ruschendorf2024model}
\APACinsertmetastar {%
ruschendorf2024model}%
\begin{APACrefauthors}%
R{\"u}schendorf, L.%
, Vanduffel, S.%
\BCBL {} Bernard, C.%
\end{APACrefauthors}%
\unskip\
\newblock
\APACrefYear{2024}.
\newblock
\APACrefbtitle {Model Risk Management: Risk Bounds Under Uncertainty} {Model risk management: Risk bounds under uncertainty}.
\newblock
\APACaddressPublisher{}{Cambridge University Press}.
\PrintBackRefs{\CurrentBib}

\bibitem [\protect \citeauthoryear {%
Sadooghi-Alvandi%
, Shishebor%
\BCBL {}\ \BBA {} Mardani-Fard%
}{%
Sadooghi-Alvandi%
\ \protect \BOthers {.}}{%
{\protect \APACyear {2013}}%
}]{%
sadooghi2013sharp}
\APACinsertmetastar {%
sadooghi2013sharp}%
\begin{APACrefauthors}%
Sadooghi-Alvandi, S.%
, Shishebor, Z.%
\BCBL {} Mardani-Fard, H.%
\end{APACrefauthors}%
\unskip\
\newblock
\APACrefYearMonthDay{2013}{}{}.
\newblock
{\BBOQ}\APACrefatitle {Sharp bounds on a class of copulas with known values at several points} {Sharp bounds on a class of copulas with known values at several points}.{\BBCQ}
\newblock
\APACjournalVolNumPages{Communications in Statistics-Theory and Methods}{42}{12}{2215--2228,}
\newblock

\newblock

\PrintBackRefs{\CurrentBib}

\bibitem [\protect \citeauthoryear {%
Shemyakin%
\ \BBA {} Youn%
}{%
Shemyakin%
\ \BBA {} Youn%
}{%
{\protect \APACyear {2006}}%
}]{%
shemyakin2006copula}
\APACinsertmetastar {%
shemyakin2006copula}%
\begin{APACrefauthors}%
Shemyakin, E.A.%
\BCBT {}\ \BBA {} Youn, H.%
\end{APACrefauthors}%
\unskip\
\newblock
\APACrefYearMonthDay{2006}{}{}.
\newblock
{\BBOQ}\APACrefatitle {Copula models of joint last survivor analysis} {Copula models of joint last survivor analysis}.{\BBCQ}
\newblock
\APACjournalVolNumPages{Applied Stochastic Models in Business and Industry}{22}{2}{211--224,}
\newblock

\newblock

\PrintBackRefs{\CurrentBib}

\bibitem [\protect \citeauthoryear {%
Spreeuw%
}{%
Spreeuw%
}{%
{\protect \APACyear {2006}}%
}]{%
spreeuw2006types}
\APACinsertmetastar {%
spreeuw2006types}%
\begin{APACrefauthors}%
Spreeuw, J.%
\end{APACrefauthors}%
\unskip\
\newblock
\APACrefYearMonthDay{2006}{}{}.
\newblock
{\BBOQ}\APACrefatitle {Types of dependence and time-dependent association between two lifetimes in single parameter copula models} {Types of dependence and time-dependent association between two lifetimes in single parameter copula models}.{\BBCQ}
\newblock
\APACjournalVolNumPages{Scandinavian Actuarial Journal}{2006}{5}{286--309,}
\newblock

\newblock

\PrintBackRefs{\CurrentBib}

\bibitem [\protect \citeauthoryear {%
Spreeuw%
\ \BBA {} Owadally%
}{%
Spreeuw%
\ \BBA {} Owadally%
}{%
{\protect \APACyear {2013}}%
}]{%
spreeuw2013investigating}
\APACinsertmetastar {%
spreeuw2013investigating}%
\begin{APACrefauthors}%
Spreeuw, J.%
\BCBT {}\ \BBA {} Owadally, I.%
\end{APACrefauthors}%
\unskip\
\newblock
\APACrefYearMonthDay{2013}{}{}.
\newblock
{\BBOQ}\APACrefatitle {Investigating the broken-heart effect: a model for short-term dependence between the remaining lifetimes of joint lives} {Investigating the broken-heart effect: a model for short-term dependence between the remaining lifetimes of joint lives}.{\BBCQ}
\newblock
\APACjournalVolNumPages{Annals of Actuarial Science}{7}{2}{236--257,}
\newblock

\newblock

\PrintBackRefs{\CurrentBib}

\bibitem [\protect \citeauthoryear {%
Tankov%
}{%
Tankov%
}{%
{\protect \APACyear {2011}}%
}]{%
tankov2011improved}
\APACinsertmetastar {%
tankov2011improved}%
\begin{APACrefauthors}%
Tankov, P.%
\end{APACrefauthors}%
\unskip\
\newblock
\APACrefYearMonthDay{2011}{}{}.
\newblock
{\BBOQ}\APACrefatitle {Improved Fr{\'e}chet bounds and model-free pricing of multi-asset options} {Improved fr{\'e}chet bounds and model-free pricing of multi-asset options}.{\BBCQ}
\newblock
\APACjournalVolNumPages{Journal of Applied Probability}{48}{2}{389--403,}
\newblock

\newblock

\PrintBackRefs{\CurrentBib}

\bibitem [\protect \citeauthoryear {%
Wang%
, Young%
\BCBL {}\ \BBA {} Panjer%
}{%
Wang%
\ \protect \BOthers {.}}{%
{\protect \APACyear {1997}}%
}]{%
wang1997axiomatic}
\APACinsertmetastar {%
wang1997axiomatic}%
\begin{APACrefauthors}%
Wang, S.S.%
, Young, V.R.%
\BCBL {} Panjer, H.H.%
\end{APACrefauthors}%
\unskip\
\newblock
\APACrefYearMonthDay{1997}{}{}.
\newblock
{\BBOQ}\APACrefatitle {Axiomatic characterization of insurance prices} {Axiomatic characterization of insurance prices}.{\BBCQ}
\newblock
\APACjournalVolNumPages{Insurance: Mathematics and economics}{21}{2}{173--183,}
\newblock

\newblock

\PrintBackRefs{\CurrentBib}

\bibitem [\protect \citeauthoryear {%
Youn%
\ \BBA {} Shemyakin%
}{%
Youn%
\ \BBA {} Shemyakin%
}{%
{\protect \APACyear {1999}}%
}]{%
youn1999statistical}
\APACinsertmetastar {%
youn1999statistical}%
\begin{APACrefauthors}%
Youn, H.%
\BCBT {}\ \BBA {} Shemyakin, A.%
\end{APACrefauthors}%
\unskip\
\newblock
\APACrefYearMonthDay{1999}{}{}.
\newblock
{\BBOQ}\APACrefatitle {Statistical aspects of joint life insurance pricing} {Statistical aspects of joint life insurance pricing}.{\BBCQ}
\newblock
 \APACrefbtitle {1999 Proceedings of the Business and Statistics Section of the American Statistical Association} {1999 proceedings of the business and statistics section of the american statistical association}\ (\BPGS\ 34--38).
\newblock
\APACaddressPublisher{}{American Statistical Association Washington, DC}.
\PrintBackRefs{\CurrentBib}

\bibitem [\protect \citeauthoryear {%
Youn%
\ \BBA {} Shemyakin%
}{%
Youn%
\ \BBA {} Shemyakin%
}{%
{\protect \APACyear {2001}}%
}]{%
youn2001pricing}
\APACinsertmetastar {%
youn2001pricing}%
\begin{APACrefauthors}%
Youn, H.%
\BCBT {}\ \BBA {} Shemyakin, A.%
\end{APACrefauthors}%
\unskip\
\newblock
\APACrefYearMonthDay{2001}{}{}.
\newblock
{\BBOQ}\APACrefatitle {Pricing practices for joint last survivor insurance} {Pricing practices for joint last survivor insurance}.{\BBCQ}
\newblock
\APACjournalVolNumPages{Actuarial Research Clearing House}{1}{}{1--13,}
\newblock

\newblock

\PrintBackRefs{\CurrentBib}

\end{thebibliography}

\end{document}